\documentclass[twoside]{article}

\usepackage{algorithm}
\usepackage{algorithmic}
\usepackage{amsfonts}
\usepackage{amsmath}
\usepackage{amsthm}
\usepackage{bbm}
\usepackage{graphicx}
\usepackage{xcolor}
\usepackage{hyperref}
\usepackage[a4paper]{geometry}
\newcommand{\E}{\mathbb E}
\newcommand{\N}{\mathbb N}
\renewcommand{\P}{\mathbb P}
\newcommand{\R}{\mathbb R}
\newcommand{\1}{\mathbbm 1}
\newcommand{\Var}{\operatorname{Var}}

\newtheorem{theorem}{Theorem}[section]
\newtheorem{lemma}[theorem]{Lemma}

\theoremstyle{definition}
\newtheorem{example}[theorem]{Example}
\newtheorem{remark}[theorem]{Remark}

\usepackage{amsmath}

\DeclareMathOperator*{\argmin}{arg\,min}
\newcommand{\figurewidth}{\columnwidth}

\begin{document}

\title{Federated Bayesian Computation via Piecewise Deterministic Markov Processes}

\author{Joris Bierkens\footnote{Delft University of Technology}, Andrew B. Duncan\footnote{Imperial College London, Alan Turing Institute}}

\maketitle

\begin{abstract}
  When performing Bayesian computations in practice, one is often faced with the challenge that the constituent model components and/or the data are only available in a distributed fashion, e.g. due to privacy concerns or sheer volume.   While various methods have been proposed for performing posterior inference in such federated settings, these either make very strong assumptions on the data and/or model or otherwise introduce significant bias when the local posteriors are combined to form an approximation of the target posterior. By leveraging recently developed methods for Markov Chain Monte Carlo (MCMC) based on Piecewise Deterministic Markov Processes (PDMPs), we develop a computation- and communication- efficient family of posterior inference algorithms (Fed-PDMC) which provides asymptotically exact approximations of the full posterior over a large class of Bayesian models, allowing heterogenous model and data contributions from each client.  We show that communication between clients and the server preserves the privacy of the individual data sources by establishing differential privacy guarantees.  We quantify the performance of Fed-PDMC over a class of illustrative analytical case-studies and demonstrate its efficacy on a number of synthetic examples along with realistic Bayesian computation benchmarks.
\end{abstract}

\section{Introduction}

In the problem of Federated Bayesian learning we are faced with the unique challenge that, either due to privacy or scalability, the model and its data are distributed across a federation of workers. In this setting, the model and/or data owned by the individual worker must not be disclosed to the other workers, and neither to any coordinator. While this problem has been studied in previous works, much of the proposed methodology has involved sacrificing the asymptotic exactness, which is characteristic of MCMC-based sampling algorithms, to facilitate the distribution and federation of the data, or alternatively being focused on a very narrow class of models. 

The problem of federated learning has largely been studied in the optimisation setting, where a global loss function which can be decomposed into local worker contributions must be optimised.   Classical strategies including Fed-SGD and Fed-Avg  \cite{mcmahan2017communication}, FedAc \cite{yuan2020federated}, and subsequent extensions, build on the general idea that each worker seeks to locally move towards the optimiser of the model based on their own local data contribution, computing a local gradient, and a central server aggregates these local gradient candidates in an appropriate fashion to obtain an estimate of the global minimiser. 

The validity of these approaches typically hinge on strong assumptions.  Firstly, it is assumed that the data across workers is independent and identically distributed, though a number of more recent federated learning methods seek to weaken these requirements, e.g. through knowledge distillation \cite{zhu2021data} or Bayesian non-parametric modelling, e.g. \cite{yurochkin2019bayesian}.  Secondly, it is assumed that workers use the same local model, though recent work on model personalization has suggested some strategies to address this \cite{mansour2020three}.

In this work, we address the problem of federated learning in a Bayesian setting, i.e. we seek to generate samples from a global posterior probability distribution obtained as a multiplicative composition of local posteriors distributed across the workers, and without sharing of model and/or data.   This is an inherently more challenging problem due to the fact that far more information pertaining to the local posterior distribution must be somehow communicated with other workers while ensuring privacy of data/model.    

Previous works have sought to lift methodology from the federated learning to the Bayesian setting, employing Stochastic Gradient Langevin Dynamics (SGLD)-based generalisations of Federated Learning counterparts, e.g. \cite{zhang2019cyclical, Elmekkaoui2021, vono2022qlsd}. Similarly, the Langevin-type algorithm proposed in \cite{sun2022federated}  combines distributed MCMC with compression techniques to reduce the  burden of communicating large gradients. 

Related approaches seek to employ approximations of the local posterior contributions which are used to communicate information to the central server.  Such approaches include distributed variational inference \cite{zhang2018advances, hasenclever2017distributed}, using Gaussian approximations \cite{al2020federated} and ensemble approaches \cite{linsner2021approaches}.   In \cite{bhattBayesianFederatedLearning2022}, local predictive posterior contributions are distilled and stored into a neural network which is communicated with the central server.    

Some works seek to reformulate the federated learning problem through the lens of Bayesian model averaging, where the local model contributions are combined into an accurate global approximation as a model ensemble \cite{chen2020fedbe, thorgeirssonProbabilisticPredictionsFederated2020}, building on other Bayesian uncertainty quantification methods used in deep learning such as \cite{maddox2019simple}.   Related to this are approaches which adopt a Bayesian hierarchical modelling view of Federated Learning, introducing hierarchical priors and fixed and random effects to share global information across the different federated workers, \cite{kotelevskii2022fedpop}.

All of the approaches discussed above either employ local posterior approximations to enable effective communication, and/or are contingent on very strict approximations on the structure of the model.  To our knowledge, there is no approach which can perform a full, (asymptotically) exact Bayesian inference in this context for a general Bayesian model.  In this paper we provide a federated (or distributed) approach to Markov Chain Monte Carlo with the following properties:
(i) the correct posterior distribution is retained;
(ii) the observational data may be distributed among workers with no requirement to exchange information other than the algorithmic output;
(iii) the observational data amongst different workers does not have to be identically distributed, nor do the local prior distributions have to be the same; 
(iv) the efficiency of the federated approach compares favourably to the non-federated approach in the sense that the algorithmic slowing down is compensated by the fact that computation is distributed among workers;
(v) the amount of information that is communicated between the workers and the server respects the privacy requirements of the worker, which can be quantified from a differential privacy viewpoint.

We will base our approach on the framework of Piecewise Deterministic Monte Carlo \cite{BouchardCoteVollmerDoucet2017,BierkensFearnheadRoberts2016}, which we will introduce in Section~\ref{sec:PDMC}. As discussed in Section~\ref{sec:federated-PDMC} this framework can be easily extended to allow for a federated (or distributed) approach while retaining the correct stationary distribution. The computational efficiency of our method is discussed in Section~\ref{sec:efficiency}. We will also consider our approach from the viewpoint of differential privacy in Section~\ref{sec:differential-privacy}. We provide numerical experiments for several examples to establish proof of concept and investigate efficiency properties in Section~\ref{sec:numerics}.

\section{Piecewise Deterministic Monte Carlo}
\label{sec:PDMC}
Here we concisely describe the framework of Piecewise Deterministic Monte Carlo (PDMC) in some generality.
Essentially, a PDMC sampler is based on a piecewise deterministic Markov process. This is a continuous time Markov process which moves along continuous deterministic trajectories, until at random times, a jump within the state space is made. In PDMC the state space consists of a position process $X(t) \in \R^d$ and a velocity process $V(t)$ taking values in a set $\mathcal V \subset \R^d$. The jumps (or \emph{events}) will only affect the velocity. The process will be designed to have a particular stationary probability distribution $\mu(dx,dv)$ with marginal position distribution $\pi(dx) = \int_{v \in \mathcal V} \mu(dx,dv)$. Here $\pi$ may be considered to be a Bayesian posterior distribution of interest.

\subsection{Deterministic dynamics}
In the general setting, the deterministic dynamics are described as the solution of an ordinary differential equation
\begin{equation} \label{eq:ode} \frac{d x(t)}{dt} = v(t), \quad \frac{d v(t)}{dt} = \psi(x(t)),\end{equation}
where $\psi : \R^d \rightarrow \R$ is a sufficiently regular function so that solutions to~\eqref{eq:ode} are defined uniquely, e.g., $\psi$ may be assumed to be globally Lipschitz. The deterministic dynamics are assumed to preserve a `reference' stationary measure $\mu_0(dx,dv) =  \pi_0( d x) \otimes \nu(dv)$, where $\pi_0(dx) = \exp(-U_0(x))\, dx$ for a suitable function $U_0:\R^d \rightarrow \R$ and $\nu$ is a probability measure on $\mathcal V$. This means that for a solution $\phi(t;x_0,v_0) := (x(t;x_0,v_0), v(t;x_0,v_0))$ to~\eqref{eq:ode} with initial condition $(x_0,v_0)$, we have for all integrable  $f : \R^d \times \mathcal V \rightarrow \R$ that
\begin{align*} & \int_{x_0 \in \R^d, v_0 \in \mathcal V} f(\phi(t;x_0,v_0)) \mu_0(dx_0, d v_0)\\ & = \int_{x \in \R^d, v \in \mathcal V} f(x,v) \, \mu_0(dx,dv).\end{align*} An interesting special case is when $\pi_0(dx) $ is chosen to be the prior distribution in a Bayesian inference problem, but this is not necessary.

\begin{example}[Zig-Zag Sampler]
For the Zig-Zag Sampler (ZZS, \cite{BierkensFearnheadRoberts2016}), we take $\psi(x) = 0$, $U_0(x) = 0$, $\mathcal V = \{-1,+1\}^d$ and the stationary velocity distribution is taken to be $\nu = \text{Uniform}(\mathcal V)$. We see that the velocities assume only discrete values which do not change under the deterministic dynamics.
\end{example}

\begin{example}[Bouncy Particle Sampler and Boomerang Sampler]
Let $\mathcal V = \R^d$ equipped with a Gaussian measure $\nu = \mathcal N(0,\Sigma)$. 
where $\Sigma$ is a positive definite matrix. The dynamics~\eqref{eq:ode} preserve $\mu_0$ by taking $\psi(x) = -\Sigma \nabla U_0(x)$.
In particular, for the  Bouncy Particle Sampler (BPS, \cite{BouchardCoteVollmerDoucet2017} we take $U_0(x) = 0$ and thus $\psi(x) = 0$, and usually $\Sigma = I_d$ so that $\mu_0(dx,dv)  = \text{Leb}(\R^d)(dx)\footnote{$\text{Leb}(\R^d)$ denotes Lebesgue measure on $\R^d$.} \otimes \mathcal N(0, I_d)(dv)$.  For the Boomerang Sampler \cite{Bierkens2020a}, we take $U_0(x) = \tfrac 1 2 x^T \Sigma^{-1}x$, so $\psi(x) = -x$ and have $\mu_0(dx,dv) = \mathcal N(0, \Sigma)(dx) \otimes \mathcal N(0, \Sigma)(dv)$.
In contrast to the ZZS, the BPS has a continuous space of possible velocities, but as for ZZS, the velocities do not change under the deterministic dynamics. For the Boomerang Sampler the deterministic dynamics correspond to a (skewed) harmonic oscillator.
\end{example}

\subsection{Jumps}
Next we specify the jumping mechanism which changes the velocity at random times. This is governed by a jump intensity $\lambda(x,v)$ and a Markov jump kernel  $Q(x,v, dv') : \R^d \times \mathcal V \times \mathcal B(\mathcal V)\footnote{$\mathcal B(\mathcal V)$ denotes the $\sigma$-field of Borel subsets of $\mathcal V$.} \rightarrow [0,1]$. More generally we may have multiple types of jumps with multiple types of rates $(\lambda_i)_{i=1}^k$ and jump distributions $(Q_i)_{i=1}^k$, competing for which event occurs first. Suppose we start from time $0$ at position $(x_0,v_0)$ and recall that we have deterministic dynamics $t \mapsto \phi(t;x_0,v_0)$. The distribution of the inter-event times $\tau_i$ are given by
\[ \P(\tau_i \ge t) = \exp \left( -\int_0^t \lambda_i(\phi(s;x_0,v_0)) \, d s \right).\]
The event that actually takes place is specified by setting $i_0 = \argmin_i \tau_i$. At time $\tau_{i_0}$ we make a transition according to the selected jump kernel $Q_{i_0}$, so that the distribution of the velocity after the jump is given by
$Q_{i_0}(\phi(\tau_{i_0};x_0,v_0), \cdot)$.
\\

\begin{remark}
It is always possible to write a combination of jump mechanisms $(\lambda_i, Q_i)_{i=1}^k$ as a single jump mechanism $(\lambda, Q)$ by defining
\begin{align*} 
\lambda(x,v) & = \sum_{i=1}^k \lambda_i(x,v), \\
Q(x,v,dv') & = \sum_{i=1}^k \frac{\lambda_i(x,v)}{\lambda(x,v)}\1_{\lambda(x,v) > 0} Q_i(x,v,dv').\end{align*}
This provides a convenient notational simplification which we will use whenever this does not cause confusion.
\end{remark}

Using the notation of the previous remark,  Algorithm~\ref{alg:PDMC} describes a general PDMC sampler.

\begin{algorithm}[ht!]
\renewcommand{\algorithmicrequire}{\textbf{Input:}}
\renewcommand{\algorithmicensure}{\textbf{Output:}}

\caption{Piecewise Deterministic Monte Carlo}
\begin{algorithmic}[1]
\REQUIRE Initial condition $(x, v) \in \R^d \times \mathcal V$.\\
\ENSURE The sequence of skeleton points $(T_k, X_k, V_k)_{k=0}^{\infty}$.
\STATE Set $(T_0, X_0, V_0)=(0, x,v)$.
\FOR{$k=0,1,2,\ldots$ (until a stopping criterion is met)} 
	\STATE{Simulate $\tau$ such that \begin{equation} \label{eq:switching-time}\P(\tau \geq t) = \exp \left( -\int_0^t \lambda(\phi(s;X_k, V_k)) \ d s \right)\end{equation}}
	\STATE{Set \begin{align*}T_{k+1} &= T_{k} + \tau,\\ 
							 (X_{k+1},\tilde V_{k+1}) &= \phi(\tau;X_k, V_k)\end{align*}}
    \STATE{Simulate $V(T_{k+1}) \sim Q(X_{k+1}, \tilde V_{k+1}, \cdot)$}
\ENDFOR
\end{algorithmic}
\label{alg:PDMC}
\end{algorithm}

In practice it may be challenging to simulate $\tau$ satisfying~\eqref{eq:switching-time}. We  discuss the usual approach of Poisson thinning in the Appendix.

\subsection{Stationary distribution}
\label{sec:stationary-distribution}
It is possible to formulate conditions on $\lambda_i$ and $Q_i$ in order to have a prespecified stationary distribution. For a function $f : \R^d \times \mathcal V \rightarrow \R$ we write
\[ Qf(x,v) = \int_{v'} Q(x,v, dv') f(x,v') \]
Suppose we wish to have the distribution $\exp(-U(x)) \pi_0(dx)$ as (marginal) stationary distribution.
In order to achieve this we impose the following conditions (understood to hold for all bounded measurable $f : \R^d \times \mathcal V \rightarrow \R$):
\begin{itemize}
    \item[(i)] Invariance of $\nu$ under the jump kernels: for all $x \in \R^d$,
    \begin{equation}\int_{v \in \mathcal V} Q_i  f(x,v) \, \nu(dv) = \int_{v \in \mathcal V} f(x,v) \, \nu(dv) \end{equation}
    \item[(ii)] Effective sign reversal under jumps:  for all $i = 1, \dots, k$ and $x \in \R^d$,
\begin{align} \label{eq:jump-condition}
\nonumber & \int_{v \in \mathcal V} \lambda_i(x,v) Q_i f(x,v) \, \nu(dv)\\
& = \int_{v \in \mathcal V} \lambda_i(x,-v) f(x,v) \, \nu(dv),
\end{align}
and
\item[(iii)] Event intensity condition: for all $(x,v) \in \R^d \times \mathcal V$,
\begin{equation}
\label{eq:switching-rate-condition}
\sum_{i=1}^k [\lambda_i(x,v) - \lambda_i(x,-v)] = \langle v, \nabla U(x) \rangle.
\end{equation}

\end{itemize}

Under the stated conditions it follows that the process with deterministic dynamics~\eqref{eq:ode}, and jumps according to $(\lambda_i, Q_i)_{i=1}^m$, has stationary distribution $\mu(dx,dv) \propto \exp(-U(x)) \, \pi_0(d x) \otimes \nu(dv)$.
The proof of this result depends on the notion of the Markov process generator and is beyond the scope of this work; see e.g. \cite{BierkensFearnheadRoberts2016,BouchardCoteVollmerDoucet2017,Bierkens2020a}

\begin{example}[Zig-Zag Sampler]
For Zig-Zag, we consider for $i = 1, \dots, d$,  $\lambda_i(x,v) = \max(v_i \partial_i U(x), 0) + \gamma_i(x)$ and  $Q_i f(x,v) = f(x,F_iv)$, where $\gamma_i$ is a non-negative function (called the \emph{excess switching rate} or \emph{refreshment rate}
\[ (F_i v)_j = \begin{cases} v_j \quad & j \ne i, \\
- v_j \quad & j = i. \end{cases}\]
This corresponds to flipping the $i$th direction of the velocity at a rate which depends on the $i$th partial derivative of $U$ as indicated.
\end{example}

\begin{example}[Bouncy Particle Sampler and Boomerang Sampler]
For BPS and Boomerang, recall that $\nu(dv) \propto \mathcal N(0, \Sigma)$ for a positive definite matrix $\Sigma$. We  take
\[ \lambda(x,v) = \max(\langle v, \nabla U(x) \rangle, 0)  \]
and $Q f(x,v) = f(x,R(x)v)$, where
\[ R(x)v = v - 2 \frac{\langle v,  \nabla U(x) \rangle}{|\Sigma^{1/2} \nabla U(x)|^2} \Sigma \nabla U(x).\]
This corresponds to a \emph{reflection} of the velocity in the direction of the gradient of $U$. In addition, we require a \emph{refreshment} jump at rate $\lambda_0(x)$, which independently draws a new velocity from the distribution $\nu$: without this refreshment the process will in general not be ergodic, i.e., it will not explore the full state space.
\end{example}

\subsection{The output of a PDMC algorithm}
\label{sec:PDMC-output}

In order to determine the full continuous time trajectory, it is sufficient to determine the positions and velocities $(X_k, V_k)$ immediately after jumps. These points are called the \emph{skeleton points}. The continuous time trajectory is obtained by the deterministic dynamics originating from the skeleton points, as
\[ (X(t),V(t)) = \phi(t - T_k; X_k, V_k), \quad T_k \le t < T_{k+1}.\] 

Provided the piecewise deterministic process is ergodic (as discussed in e.g. \cite{BouchardCoteVollmerDoucet2017, BiRoZi2019} for BPS and Zig-Zag respectively) we have the following approximation for our computation of interest: with probability one, 
\[ \int_{\R} h(x) \ \pi(x) \ d x = \lim_{t \rightarrow \infty} \frac 1 t \int_0^t h(X(s)) \ d s,\]
were $(X(t),V(t))$ is any random realization of the piecewise deterministic process with characteristics $(\phi, \lambda,Q)$ and arbitrary initial condition. Due to the piecewise linear nature of the trajectories of $(X(t))_{t \geq 0}$ it is often very straightforward to evaluate the one-dimensional integrals in this expression. Alternatively, one can obtain a discrete set of samples $(\widetilde X_k)_{k \in \N}$ by setting $(\widetilde X_k, \widetilde V_k) = (X(k \delta), V(k \delta))$, for some arbitrary $\delta > 0$. In this case the usual MCMC approximation formula 
\[ \int_{\R} h(x) \ \pi(x) \ d x = \lim_{K \rightarrow \infty} \frac 1 K \sum_{k=1}^K h(\widetilde X_k) \]
is satisfied with probability one, because $(\widetilde X_k, \widetilde V_k)$ can be seen as a discrete time ergodic Markov chain in $\R^d \times \mathcal V$ with marginal invariant density on $\R^d$ given by $\pi$.

\section{Federated Piecewise Deterministic Monte Carlo}
\label{sec:federated-PDMC}

Now consider the setting in which $\pi(dx)$ admits the factorization 
\[ \pi(dx) \propto \exp \left( -U_0(x) - \sum_{m=1}^M U_m(x) \right) \, dx.\]
We will distribute the simulation of $\pi$ over $M$ workers, where we assume that the function $U_0$ (with its gradient) is available to every worker, whereas for each $m = 1, \dots M$ the function $U_m$ (with its gradient) is only available to the $m$th worker.

\subsection{Federated computation of the first event}
The essential idea of Federated PDMC is that every worker $m \in \{1, \dots, M\}$ proposes a switching time associated to their own rate function $\lambda_m$. This proposed  switching time is communicated (along with the proposed change in velocity) to the coordinating server which selects the minimum of the proposed switching times and the proposed switch. From this time and the new combination of position and velocity, the process is repeated. Under simple conditions this approach can be seen to have the correct stationary distribution.

We suppose that every worker has its own jump mechanism consisting of jump intensity $\lambda_m(x,v)$ and jump kernel $Q_m(x,v,dv')$, satisfying~\eqref{eq:jump-condition} (replacing $i$ by $m$)
and $\lambda_m(x,v) - \lambda_m(x,-v) = \langle v, \nabla U_m(x) \rangle$.
Given initial condition $(x,v) \in \R^d \times \mathcal V$ every worker computes the first switching time $\tau_m$ according to the rate $\lambda_m(x,v)$, i.e.
\[ \P_{x,v}(\tau_m \geq t) = \exp \left( - \int_0^t \lambda_m(\phi(s;x, v)) \ d s \right).\]
(The deterministic dynamics $\phi$ are identical for all workers.)
Furthermore every worker computes a new choice of velocity $v_m$ according to their individual jump distribution $Q_m$, i.e. every worker simulates
\[ \tilde V^{(m)} \sim Q_m(\phi(\tau_m; x,v), \cdot).\]
Then every worker sends its proposal $(\tau_m, \tilde V^{(m)})$ to the server. The server determines the minimum switching time and associated new velocity.
Therefore the effective switching time for Federated PDMC is the first arrival time of an inhomogeneous Poisson process with rate $\lambda_{\text{fed}}(x,v) = \sum_{m=1}^M \lambda_m(x,v)$.
The associated effective jump kernel is
\[ Q_{\text{fed}}(x,v, dv') = \sum_{m=1}^M \frac{\lambda_m(x,v)}{\lambda_\text{fed}(x,v)} \1_{\lambda_{\text{fed}}(x,v) > 0} Q_m(x,v,dv').\]

The above procedure, described in detail in Algorithm~\ref{alg:federated-pdmc},  provides a genuinely federated algorithm since every machine only requires access to $U_m$, without affecting the invariant probability distribution. Indeed it is straightforward to verify that the conditions of Section~\ref{sec:stationary-distribution} are satisfied for $(\lambda_{\text{fed}}, Q_{\text{fed}})$. Moreover any ergodicity properties of the non-federated PDMC algorithm carry over to the Federated PDMC algorithm, since the effective event rate $\lambda_{\text{fed}}$ of Federated PDMC is increased relative to non-federated PDMC, as discussed in Section~\ref{sec:effective-event-rate}. Therefore the obtained skeleton points may be used as discussed in Section~\ref{sec:PDMC-output}.

\begin{algorithm}[ht!]

\renewcommand{\algorithmicrequire}{\textbf{Input:}}
\renewcommand{\algorithmicensure}{\textbf{Output:}}

\caption{Federated Piecewise Deterministic Monte Carlo}

\begin{algorithmic}[1]
\REQUIRE Initial condition $(x, v) \in \R^d \times \mathcal V$.\\
\ENSURE The sequence of \emph{skeleton points} $(T_k, X_k, V_k)_{k=0}^{\infty}$.
\STATE Set $(T_0, X_0, V_0)=(0, x, v)$.
\FOR{$k=0,1,2,\ldots$}
	\STATE{Locally: Every worker simulates $\tau_m$ and $\tilde V^{(m)}$ such that \begin{align*} \P(\tau_m \geq t) & = \exp \left( -\int_0^t \lambda_m(\phi(s;X_k,V_k)) \ d s \right)\\
 \tilde V^{(m)} & \sim Q_m(\phi(\tau_m;X_k,V_k),\cdot) \end{align*}}
	
	\STATE{Centrally: Set \begin{align*} m_0 & = \argmin \{ \tau_1, \dots, \tau_M \}, \\
	T_{k+1} &= T_{k} + \tau_{m_0},\\
    X_{k+1} &= \phi(\tau_{m_0}; X_k, V_k),\\
    V_{k+1} &= \tilde V^{(m_0)}.\end{align*}}
\ENDFOR
\end{algorithmic}

\label{alg:federated-pdmc}
\end{algorithm} 

\begin{remark}
Each worker only requires access to the component of the global variable which the function $U_m$ depends on.  It is straightforward to reformulate Algorithm \ref{alg:federated-pdmc} so that each worker only acts on a set of local variables, which are subsequently mapped onto the global state $X_n$ by the central server, and vice versa. 
\end{remark}

\begin{remark}
There is flexibility in how to accommodate for the prior distribution: it can be absorbed in the function $U_0$ or distributed amongst the workers through the functions $U_m$. Alternatively it could be handled by an artificial extra worker which simulates event times associated with the prior distribution. Also combinations of these approaches are possible.
\end{remark}

\section{Computational efficiency of Federated PDMC}
\label{sec:efficiency}

Consider a Bayesian context in which we have $N$ observations, distributed among $M$ workers, each worker having access to a batch of size $n_m$ points, with $N = \sum_{m=1}^M n_m$. In many relevant cases the computation of the proposed switching times $\tau_m$ will be a bottleneck factor for the total computational effort, and it is reasonable to expect that this effort is linear in the size of the data. Since the workers operate in parallel, we find that the computational effort required is of order $\max_{m=1,\dots,M} n_m$. As a special case, if all batch sizes are equal, i.e., $n_m = N/M$ for all $m$, the computational effort required per iteration is $N/M$. We see that the computational effort per proposed switch is reduced by a factor $M$ relative to the case in which all $N$ observations would be processed by a single worker. 

This simple computation does not yet paint the full picture: Although the invariant distribution is not affected by employing Federated PDMC, the  event rate is modified since the operation of taking the positive part $a \mapsto (a)_+$ occurs at each individual worker. We will consider the effect of this in some detail below.

\subsection{Expected event rate for the exponential family}
\label{sec:effective-event-rate}
Consider the `canonical rate' $\lambda_{\text{can}}$, corresponding to a single machine generating the switches associated with the potential function $U(x) = \sum_{m=1}^M U_m(x)$. In the context of the BPS and Boomerang Sampler (as an example),
\begin{align*} 
\lambda_{\text{fed}}(x,v) & = \sum_{m=1}^M \langle v, \nabla U_m(x) \rangle_+ \\ & \ge \left(  \sum_{m=1}^M\langle v, \nabla U_m(x) \rangle\right)_+ 
= \lambda_{\text{can}}(x,v).\end{align*}

We further investigate the expected switching rate for the federated one-dimensional Zig-Zag sampler for exponential family models indexed by a parameter $x$. We assume the data is generated from this model for a fixed (unknown) parameter value $x_0$.

For simplicity we assume a flat (Lebesgue) prior measure.

We show in the Appendix that in this situation the posterior expected event rate under the federated intensity and under the distribution of the data is  magnified by a factor $\sqrt{M}$, compared to the canonical rate:
\[ \E_{\mu} \lambda_{\text{fed}}(x,v) = \E_{\mu}\lambda_{\text{can}} (x,v) + \mathcal O(\sqrt{M N}),\]
while $\E_{\mu} \lambda_{\text{can}}(x,v) = \mathcal O(\sqrt{N})$.

\begin{remark}
\label{rem:diffusion}
The larger (expected) effective switching rate results in more simulated switches per unit time interval, so that a larger computational effort is required to simulated such an interval. There is another aspect which affects computational efficiency: As the switching rate increases beyond the canonical rate, the process trajectories become more diffusive (see \cite{BierkensDuncan2016}), resulting in an increased Monte Carlo error per simulated unit time interval. The quantification of this error is beyond the scope of this work.
\end{remark}

\begin{remark} For an exponential family the analysis breaks down into easily manageable parts but it may well be possible to generalize these results beyond this setting. Also the scaling dependence on dimension in multivariate settings is left for further research.
\end{remark}

\begin{remark} The increase of the event rate $\lambda_{\text{fed}}$ compared to $\lambda_{\text{can}}$ can be reduced by taking a control variates approach, in a similar approach as discussed in \cite{BierkensFearnheadRoberts2016}.
\end{remark}

\section{Differential privacy}
\label{sec:differential-privacy}

A relevant aspect of federated inference is the amount of privacy achieved by taking a federated approach. In Federated PDMC, every worker only computes its proposed switching time along with the corresponding proposed change of velocity. This is a very limited amount of information and as such at an intuitive level we may be optimistic about the privacy achieved by the Federated PDMC approach.

For a theoretical approach, we may employ the concept of differential privacy \cite{Dwork2006}. 
In differential privacy the privacy of a stochastic algorithm is quantified as follows. For $\varepsilon, \delta > 0$, an algorithm gives \emph{$\delta$-approximate $\varepsilon$-indistinguishability} or simpliy \emph{$(\varepsilon,\delta)$-differential privacy} if for outputs $\tau$ and $\tilde \tau$ corresponding to data sets differing at most one row,
\begin{equation} \label{eq:epsdelta-differential-privacy} \P(\tau \in S) \le \exp(\varepsilon) \P(\tilde \tau \in S) + \delta\end{equation}
for any measurable set $S$.

We will now investigate the differential privacy of Federated PDMC.
We make the assumption that 
\[ \pi(x) \propto \exp \left( - U_0(x) - \sum_{m=1}^M \sum_{i=1}^{n_m} U_i^{(m)}(x)\right),\]
where a single change of the $i$th observation in batch $m$ results in a change 
$U_i^{(m)}(x) \to \widetilde U_i^{(m)}$.

\subsection{Differential privacy for communication of switching times}

We suppress the dependence on $m$, so let $\lambda$ and $\widetilde \lambda$ denote the  switching intensities of a single machine with corresponding proposed switching times $\tau$ and $\tilde \tau$. We have the following result.

\begin{theorem}
\label{thm:epsdelta-differential-privacy-bounded-difference}
Suppose $\lambda(x,v) \ge \rho$  and $|\lambda(x,v) - \tilde \lambda(x,v)| \le K$ for for all $(x,v)$ and some constants $\rho > 0$ and $K > 1$. Then for any $S \subset [0,\infty)$ we have that~\eqref{eq:epsdelta-differential-privacy} holds,
where $\varepsilon > \log \left( 1 + \frac K \rho \right)$ and 
\[ \delta = \exp\left(- \frac{\rho}{K} \left[ \varepsilon - \log \left( 1 + \frac{K}{\rho} \right) \right]\right).\]
In particular, for $\varepsilon > 0$ and $\delta > 0$, if 
\begin{equation} \label{eq:condition-rho-differential-privacy} \rho \ge K \left( \frac{ 1 +\log (1/\delta)}{\varepsilon} \right),
\end{equation} we have that~\eqref{eq:epsdelta-differential-privacy} holds.
\end{theorem}

The proof of Theorem~\ref{thm:epsdelta-differential-privacy-bounded-difference} is provided in the Appendix.

Theorem~\ref{thm:epsdelta-differential-privacy-bounded-difference} establishes that in order to achieve a certain level of differential privacy, we may tune the \emph{refreshment rate} of the algorithm to be at least $\rho$ as specified by~\eqref{eq:condition-rho-differential-privacy}. This can be achieved by setting the switching rate for each machine $m$ as
\[ \lambda_m^{\rho}(x,v) = \lambda_m(x,v) + \rho,\]
where $\lambda_m(x,v)$ is any valid switching intensity and with $\rho$ as desired. Indeed, the conditions of Section~\ref{sec:stationary-distribution} remain to be satisfied if a constant\footnote{or more generally, a function depending on $x$ only} is added to a valid event intensity.

The following example illustrates the condition $|\lambda(x,v) - \tilde \lambda(x,v)| \le K$  of Theorem~\ref{thm:epsdelta-differential-privacy-bounded-difference}.

\begin{example}[Logistic regression]

Consider a logistic regression example, with explanatory variables $(\xi_i^{(m)})$ and binary output variable $(\eta_i^{(m)})$.  Details for this setting may be found in the Appendix.
We consider the situation where all covariates $\xi_i^{(m)}$ belong to a bounded set, i.e.,
$\|\xi_i^{(m)}\| \le K$ for all $i, m$,
and use a piecewise linear sampler with fixed velocity magnitude $\|v\| = 1$.
Suppose we change a single observation $\eta_i^{(m)}$ to $\tilde \eta_i^{(m)} = 1 - \eta_i^{(m)}$. The associated switching intensity difference then satisfies
\[ |\tilde \lambda_m(x,v) - \lambda_m(x,v)| \le | \langle \xi_i^{(m)}, v \rangle| \le K.\]

\end{example}

\section{Experiments}
\label{sec:numerics}

In this section, we demonstrate how the the Federated Zig-Zag sampler could be used for a number of classical Bayesian computation benchmarks.  In each case, we study the influence of the number of distributed workers on both the computational efficiency as well as the speed of convergence of the PDMP to the target equilibrium distribution.    Further implementation details, including the derivation of dominating Poisson processes for computing the next switch time can be found in the Appendix.  An additional example demonstrating the Federated Zig-Zag sampler on a spatial log-Gaussian Cox model is provided in the Appendix.

\subsection{ Multivariate Gaussian Distribution}

In this scenario, we assume that $N=50$ independent, identically distributed observations $y_1, \ldots, y_N$, are made of a $d=10$ dimensional multivariate Normal distribution with unobserved $\mu$ and  covariance matrix of the form $\Sigma = \alpha^2 I_{d \times d}$, where $\alpha = 1$, for simplicity.  We choose this simple example to be able to empirically validate the results of Section \ref{sec:effective-event-rate}.  Assuming a uniform improper prior on $\mu$ for simplicity, the goal is to characterise the posterior distribution of $\mu | (y_1,\ldots, y_N)$, where it is assumed that the $N$ observations are distributed equally amongst $M$ workers.   We assume that $y$ is generated with true mean $\mu_0 = (0.5, \ldots, 0.5)$.   The posterior distribution is then given by $\mathcal{N}\left(\hat{y}, \Sigma/N\right)$, where $\hat{y} = \frac{1}{N}\sum_{i=1}^N y_i$.  

In Figure \ref{fig:perf_gaussian1} we plot the estimated mean switching rate for $20$ independent runs of the Federated Zig-Zag Sampler, as a function of the number of processors, and compare them against the theoretical prediction obtained in Section~\ref{sec:effective-event-rate}.  The solid curve demonstrates the theoretically-derived expected switching rate of the central server with leading order $\sigma_0\sqrt{NM}$, where $\sigma_0 = \alpha d$.  

We can clearly see a good agreement between experiment and theory.  To better understand the effect of the model federation on the convergence of the PDMP we estimate the effective sample size (ESS) of the process across multiple, independent runs.  In Figure \ref{fig:perf_gaussian_ess2} we plot the average ESS per gradient evaluation of the full potential.  In the fully sequential setting, we observe that the ESS per evaluation is decreasing as we add workers, due to the diffusion arising from increased switching induced by the Federated Zig-Zag algorithm, as discussed in Remark \ref{rem:diffusion}.  However, if we take into account the fact that the gradient evaluations are in fact parallelised (so that effort is reduced proportionally to the number of workers),  one observes a benefit in increasing the number of workers.

\begin{figure}
\begin{center}
 \includegraphics[width=\figurewidth]{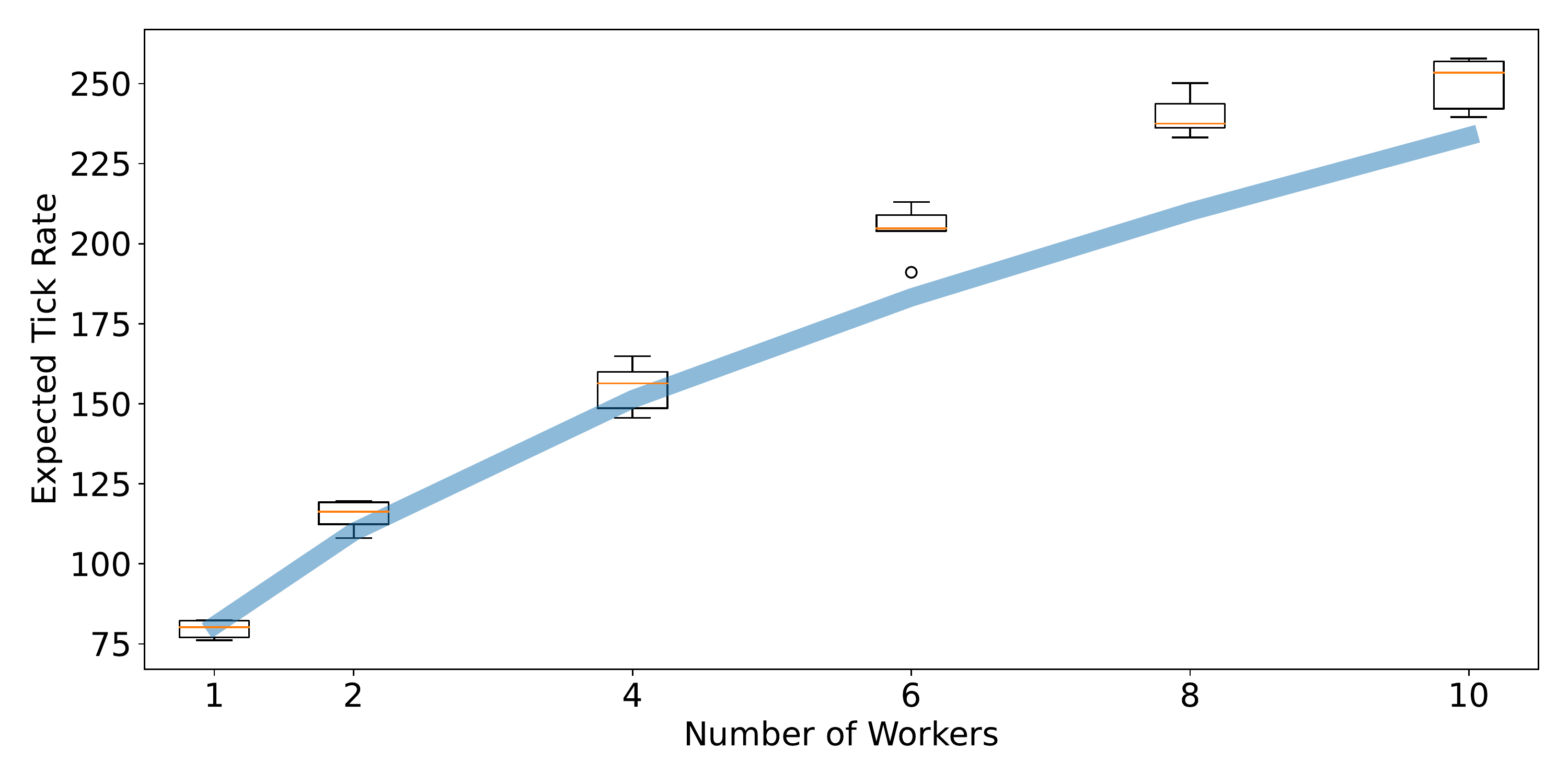}
 \caption{Effective Switching Rate for the Federated Zig-Zag Sampler for the Multivariate Gaussian Distribution compared to the theoretical expected switching rate.}
 \label{fig:perf_gaussian1}
\end{center}
\end{figure}

\begin{figure}
\begin{center}
 \includegraphics[width=\figurewidth]{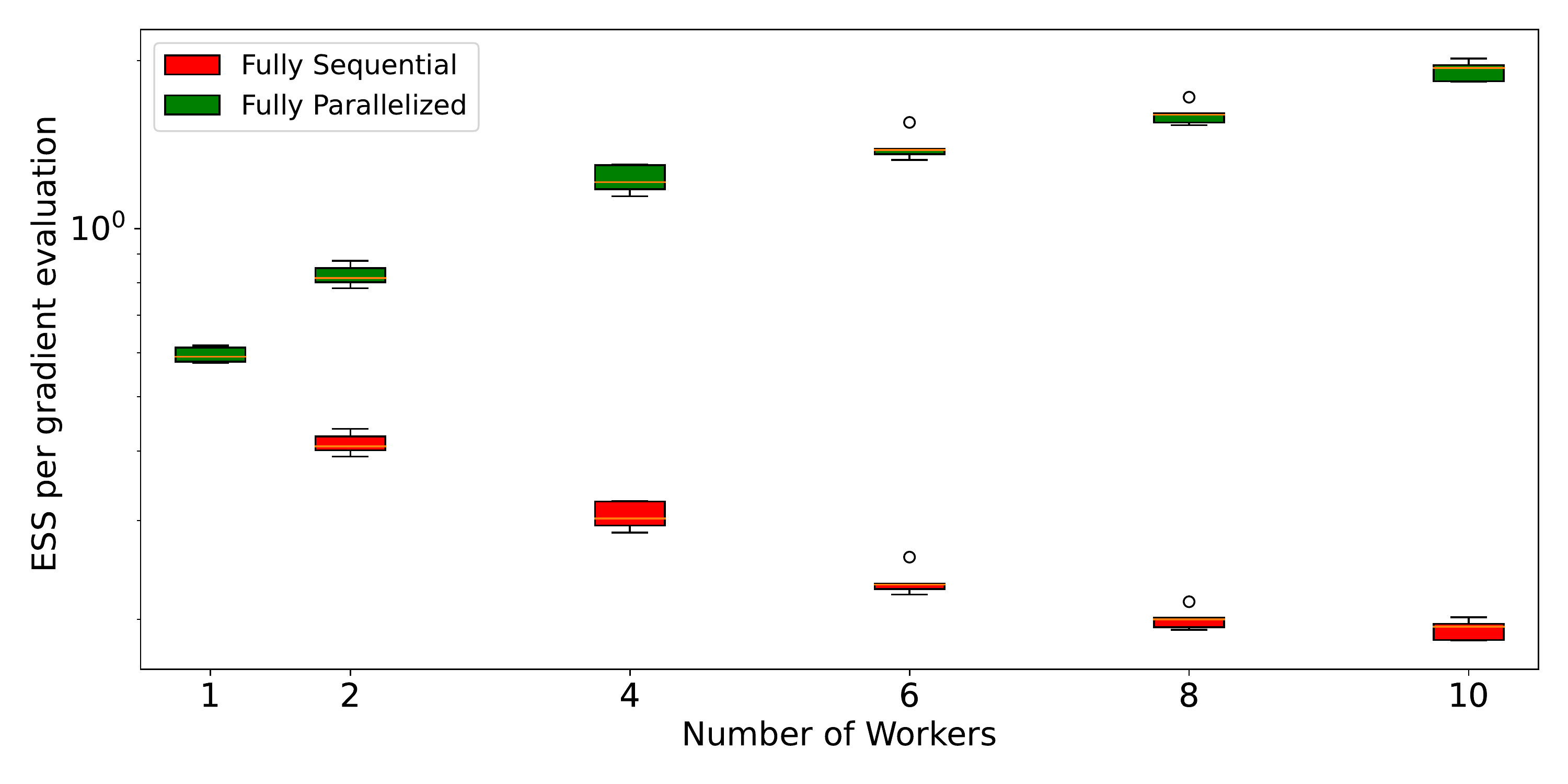}
 \caption{Effective Sample Size per  gradient evaluation of the full potential in the fully sequential and fully parallelized regimes. }
 \label{fig:perf_gaussian_ess2}
\end{center}
\end{figure}

\subsection{Logistic Regression}
We now consider a Bayesian logistic regression problem.  Given $N$ observations $\lbrace (\xi_1, \eta_1),\ldots, (\xi_N, \eta_N)\rbrace$, where $\xi_i \in \mathbb{R}^d$ and $\eta_i \in \lbrace 0, 1\rbrace$, we postulate that  $\eta_i \sim \mbox{Ber}(p_i)$,
such that $\mbox{logit}(p_i) = \xi_i^\top x,$ 
for an unknown 
$x \in \R^d$. The first component of each $\xi_i$ is taken to be equal to one to allow for an intercept in the model. The posterior distribution for $x$ given the observations is given by
$$
    \pi(x) \propto \pi_0(x)\prod_{i=1}^N \frac{\exp(\eta_i \xi_i^\top x)}{1 + \exp(\xi_i^\top x)} ,
$$
where $\pi_0$ is the prior on $x$ assumed to be standard normal, independent Gaussian distribution on all the components.   Implementation details of the implementation of the Zig-Zag sampler for this model can be found in the Appendix.   To demonstrate the Federated Zig-Zag method we generate $N = 1000$ synthetic observations with 
$d=6$, and distribute them over $M$ workers.  To show that the target posterior distribution is well approximated, we compared the output of the Federated Zig-Zag with a large MCMC sample for the same posterior, generated using Hamiltonian Monte Carlo (HMC), implemented in Blackjax \cite{blackjax2020github}.  In Figure \ref{fig:perf_lr_error} we plot the marginal 1-Wasserstein distances between $20$ independent runs of the Federated Zig-Zag scheme, run up to $T=100$, and a reference sample from the global posterior distribution.  It is clear that the overall error is small, and is not affected by the distribution of data across the workers.   

\begin{figure}
\begin{center}
 \includegraphics[width=\figurewidth]{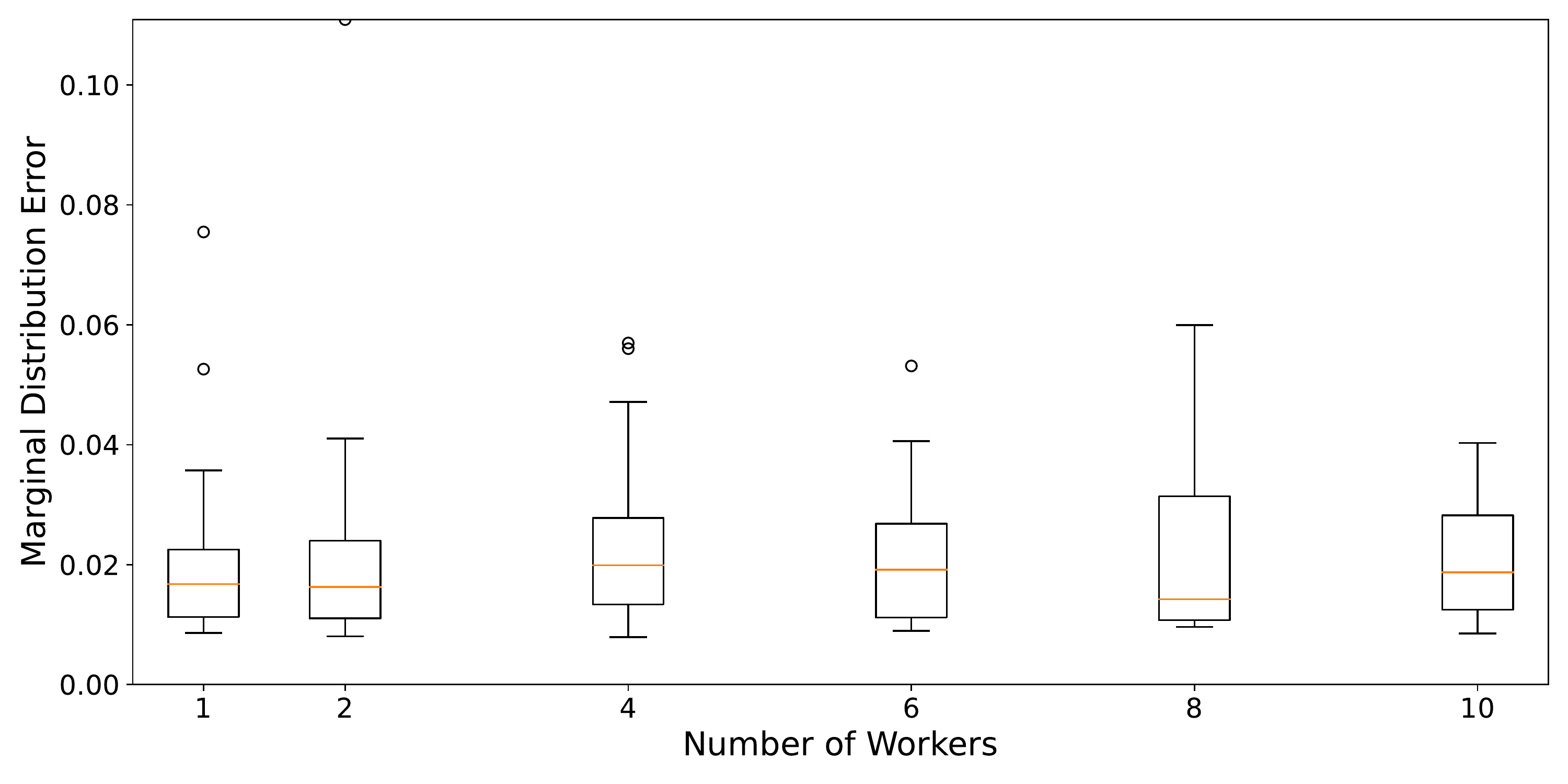}
 \caption{Marginal Wasserstein distance against a reference HMC sample for the logistic regression model.} 
 \label{fig:perf_lr_error}
\end{center}
\end{figure}

\subsection{Time Series Model}

Consider the AR(1) model in which a time series $(Y_i)$ satisfies the stochastic update rule $Y_i = c + x Y_{i-1} + \varepsilon_i$, for  $i = 1, 2, \dots$.
Here $c$ and $x$ are unknown constants, and $(\varepsilon_k)$ are i.i.d. random variables in $\R$ drawn from a distribution with density function $g$.   

We assume that we make $N$ independent observations of trajectories of the time-series, each observed at $K$ points, i.e. we observe  $y^{(1)}, \ldots, y^{(N)}$, where $y^{(i)} = (y^{(i)}_1, \ldots, y^{(i)}_K)$.   Suppose we have a joint prior density function $\pi_0(x, c, y_0)$ for $x$, $c$ and $Y_0$. For simplicity suppose that $\pi_0(x, c, y_0)$ is constant in $(x, c)$ (conditional on $y_0$). The posterior density function for $x$ and $c$ given the $N$ observed trajectories satisfies
\[ \pi(x, c) \propto \prod_{i=1}^N\prod_{k=1}^K g(y^{(i)}_k - c - x y^{(i)}_{k-1}).\]

We consider a robust inference setting, where we choose $g$ to be the density of a heavy-tailed distribution.   In this example, we choose $g$ to be a Student-T distribution with $\nu$ degrees of freedom.  Recall that the heaviness of the tails increases as $\nu \rightarrow 0$.   Mathematical details on the implementation of this model using the Zig-Zag sampler are provided in the Appendix.

To demonstrate the Federated Zig-Zag algorithm we suppose that $N$ observations are evenly distributed amongst the $M$ workers.   In Figure \ref{fig:perf_ts_error} we plot the marginal Wasserstein distance between $20$ independent runs of the Federated Zig Zag, output obtained after $T=100$ process time units of simulation, and a reference MCMC sample obtained using HMC.  We clearly see that the sampler is able to correctly approximate the correct posterior, and that this remains stable as the number of workers increases).

\begin{figure}
\begin{center}
 \includegraphics[width=\figurewidth]{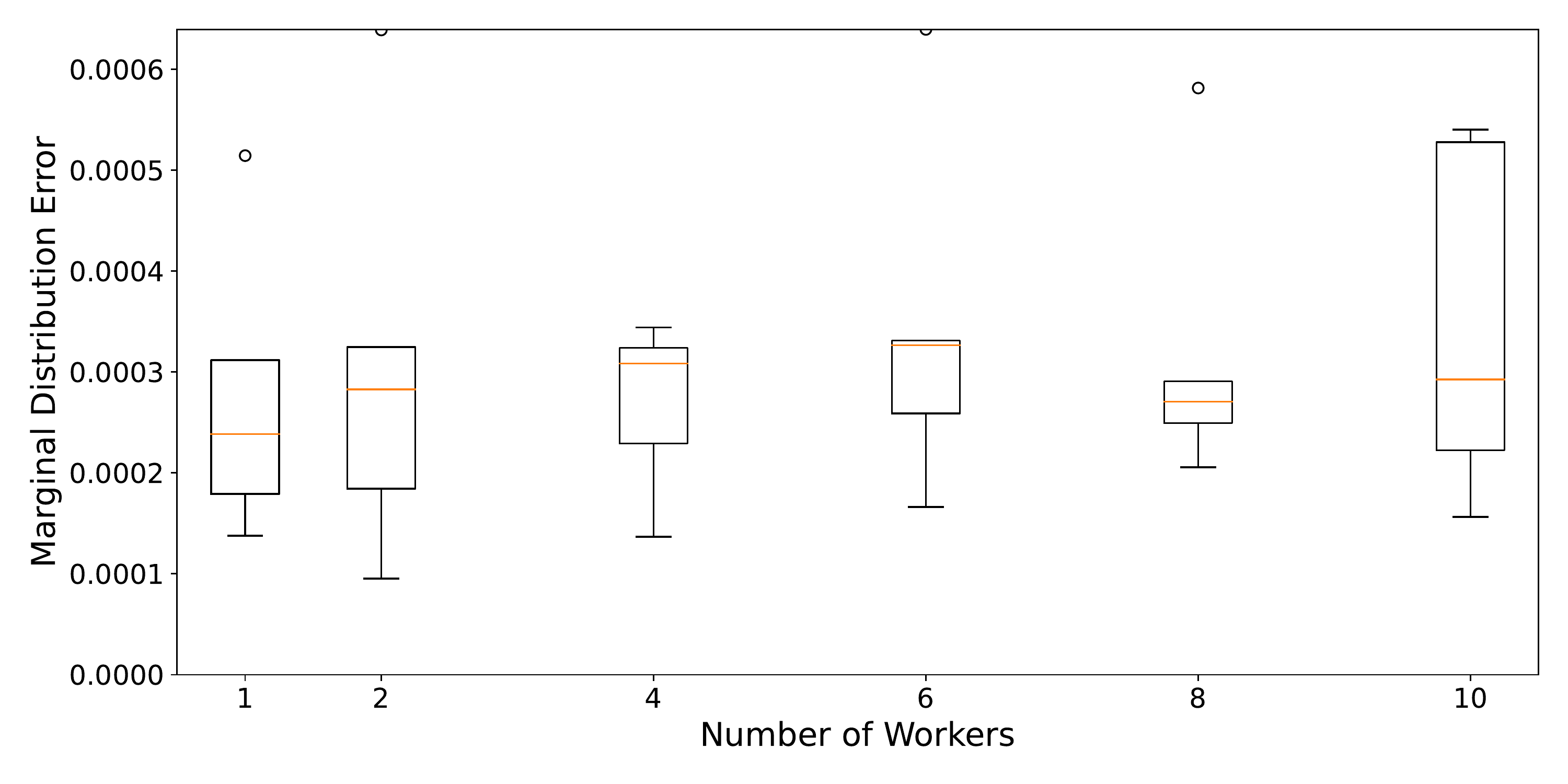}
 \caption{Marginal Wasserstein distance against a reference HMC sample for the time series model.} 

 \label{fig:perf_ts_error}
\end{center}
\end{figure}

\section{Discussion}

In this work we have introduced a generic method for Bayesian computation aimed at federated or distributed multiplicative compositions of local posterior distributions.

Our approach hinges strongly on the notion of Piecewise Deterministic Monte Carlo, a field in Bayesian computation that is currently very much under development. An intrinsic challenge of PDMC is the simulation of event times which in practice relies upon the use of a-priori bounds on the (local) gradients of the log posterior. Other numerical approaches to the computation of event times exist; see e.g. \cite{Corbella2022,Pagani2020}.

The analysis and numerics presented in this paper have focussed on a federated version of the Zig-Zag Sampler. Similar generalizations are possible for other PDMP based samplers with distinct advantages and disadvantages which we hope to study in future work.

Further research on the numerical efficiency of the distributed simulation of switching times in general multivariate settings for different PDMC algorithms is necessary. Such an understanding would be instrumental in having rules of thumb for designing optimal architectures for federated or parallel Bayesian inference, such as the optimizing the number of parallel workers. Efficiency gains may be achieved using control variates in similar spirit to \cite{BierkensFearnheadRoberts2016}.

\subsubsection*{Acknowledgements}

JB was supported by the research programme ‘Zigzagging through computational barriers’ with project number 016.Vidi.189.043, which is financed by the Dutch Research Council (NWO).
AD was supported by Wave 1 of The UKRI Strategic Priorities Fund under the EPSRC Grant EP/T001569/1 and EPSRC Grant
EP/W006022/1, particularly the ``Ecosystems of Digital Twins'' theme within those grants \& The Alan Turing Institute.

\bibliographystyle{alpha}
\bibliography{main.bib}

\newcommand{\etalchar}[1]{$^{#1}$}
\begin{thebibliography}{EMMBK21}

\bibitem[ASGXR20]{al2020federated}
Maruan Al-Shedivat, Jennifer Gillenwater, Eric Xing, and Afshin Rostamizadeh.
\newblock Federated learning via posterior averaging: A new perspective and
  practical algorithms.
\newblock {\em arXiv preprint arXiv:2010.05273}, 2020.

\bibitem[BD17]{BierkensDuncan2016}
Joris Bierkens and Andrew Duncan.
\newblock {Limit theorems for the zig-zag process}.
\newblock {\em Advances in Applied Probability}, 49(3):791--825, jul 2017.

\bibitem[BFR19]{BierkensFearnheadRoberts2016}
J.~Bierkens, P.~Fearnhead, and G.~O. Roberts.
\newblock {The Zig-Zag Process and Super-Efficient Sampling for Bayesian
  Analysis of Big Data}.
\newblock {\em Annals of Statistics}, 47(3):1288--1320, 2019.

\bibitem[BGKR20]{Bierkens2020a}
Joris Bierkens, Sebastiano Grazzi, Kengo Kamatani, and Gareth Roberts.
\newblock The {{Boomerang Sampler}}.
\newblock {\em Thirty-seventh International Conference on Machine Learning},
  June 2020.

\bibitem[BGR22]{bhattBayesianFederatedLearning2022}
Shrey Bhatt, Aishwarya Gupta, and Piyush Rai.
\newblock Bayesian {{Federated Learning}} via {{Predictive Distribution
  Distillation}}.
\newblock {\em arXiv preprint arXiv:2206.07562}, 2022.

\bibitem[BRZ19]{BiRoZi2019}
Joris Bierkens, Gareth~O Roberts, and Pierre-Andr{\'e} Zitt.
\newblock Ergodicity of the zigzag process.
\newblock {\em Ann. Appl. Probab.}, 29(4):2266--2301, 2019.

\bibitem[BVD17]{BouchardCoteVollmerDoucet2017}
Alexandre {Bouchard-C{\^o}t{\'e}}, Sebastian~J Vollmer, and Arnaud Doucet.
\newblock The {{Bouncy Particle Sampler}}: {{A Non-Reversible Rejection-Free
  Markov Chain Monte Carlo Method}}.
\newblock {\em Journal of the American Statistical Association}, 2017.

\bibitem[CC20]{chen2020fedbe}
Hong-You Chen and Wei-Lun Chao.
\newblock Fedbe: Making bayesian model ensemble applicable to federated
  learning.
\newblock {\em arXiv preprint arXiv:2009.01974}, 2020.

\bibitem[CSR22]{Corbella2022}
Alice Corbella, Simon E~F Spencer, and Gareth~O Roberts.
\newblock Automatic zig-zag sampling in practice, 2022.

\bibitem[DKM{\etalchar{+}}06]{Dwork2006}
Cynthia Dwork, Krishnaram Kenthapadi, Frank McSherry, Ilya Mironov, and Moni
  Naor.
\newblock Our data, ourselves: {{Privacy}} via distributed noise generation.
\newblock {\em Lecture Notes in Computer Science (including subseries Lecture
  Notes in Artificial Intelligence and Lecture Notes in Bioinformatics)}, 4004
  LNCS:486--503, 2006.

\bibitem[EMMBK21]{Elmekkaoui2021}
Khaoula El~Mekkaoui, Diego Mesquita, Paul Blomstedt, and Samuel Kaski.
\newblock Federated stochastic gradient {{Langevin}} dynamics.
\newblock In {\em Uncertainty in {{Artificial Intelligence}}}, pages
  1703--1712. {PMLR}, 2021.

\bibitem[Gal16]{galbraith2016event}
Nicholas Galbraith.
\newblock {\em On event-chain Monte Carlo methods}.
\newblock PhD thesis, Master’s thesis, Department of Statistics, Oxford
  University, 2016.

\bibitem[HWL{\etalchar{+}}17]{hasenclever2017distributed}
Leonard Hasenclever, Stefan Webb, Thibaut Lienart, Sebastian Vollmer, Balaji
  Lakshminarayanan, Charles Blundell, and Yee~Whye Teh.
\newblock Distributed bayesian learning with stochastic natural gradient
  expectation propagation and the posterior server.
\newblock {\em The Journal of Machine Learning Research}, 18(1):3744--3780,
  2017.

\bibitem[KVMD22]{kotelevskii2022fedpop}
Nikita Kotelevskii, Maxime Vono, Eric Moulines, and Alain Durmus.
\newblock Fedpop: A bayesian approach for personalised federated learning.
\newblock {\em arXiv preprint arXiv:2206.03611}, 2022.

\bibitem[LAD{\etalchar{+}}21]{linsner2021approaches}
Florian Linsner, Linara Adilova, Sina D{\"a}ubener, Michael Kamp, and Asja
  Fischer.
\newblock Approaches to uncertainty quantification in federated deep learning.
\newblock In {\em Joint European Conference on Machine Learning and Knowledge
  Discovery in Databases}, pages 128--145. Springer, 2021.

\bibitem[LL20]{blackjax2020github}
Junpeng Lao and R\'emi Louf.
\newblock {B}lackjax: A sampling library for {JAX}, 2020.

\bibitem[MIG{\etalchar{+}}19]{maddox2019simple}
Wesley~J Maddox, Pavel Izmailov, Timur Garipov, Dmitry~P Vetrov, and
  Andrew~Gordon Wilson.
\newblock A simple baseline for bayesian uncertainty in deep learning.
\newblock {\em Advances in Neural Information Processing Systems}, 32, 2019.

\bibitem[MMR{\etalchar{+}}17]{mcmahan2017communication}
Brendan McMahan, Eider Moore, Daniel Ramage, Seth Hampson, and Blaise~Aguera
  y~Arcas.
\newblock Communication-efficient learning of deep networks from decentralized
  data.
\newblock In {\em Artificial intelligence and statistics}, pages 1273--1282.
  PMLR, 2017.

\bibitem[MMRS20]{mansour2020three}
Yishay Mansour, Mehryar Mohri, Jae Ro, and Ananda~Theertha Suresh.
\newblock Three approaches for personalization with applications to federated
  learning.
\newblock {\em arXiv preprint arXiv:2002.10619}, 2020.

\bibitem[PCP{\etalchar{+}}20]{Pagani2020}
Filippo Pagani, Augustin Chevallier, Sam Power, Thomas House, and Simon Cotter.
\newblock {{NuZZ}}: Numerical {{Zig-Zag}} sampling for general models, 2020.

\bibitem[SSR22]{sun2022federated}
Lukang Sun, Adil Salim, and Peter Richt{\'a}rik.
\newblock Federated learning with a sampling algorithm under isoperimetry.
\newblock {\em arXiv preprint arXiv:2206.00920}, 2022.

\bibitem[TG20]{thorgeirssonProbabilisticPredictionsFederated2020}
Adam~Thor Thorgeirsson and Frank Gauterin.
\newblock Probabilistic predictions with federated learning.
\newblock {\em Entropy}, 23(1):41, 2020.

\bibitem[vdV98]{Vaart1998}
A~W van~der Vaart.
\newblock {\em Asymptotic Statistics}, volume~3.
\newblock {Cambridge University Press, Cambridge}, 1998.

\bibitem[VPD{\etalchar{+}}22]{vono2022qlsd}
Maxime Vono, Vincent Plassier, Alain Durmus, Aymeric Dieuleveut, and Eric
  Moulines.
\newblock Qlsd: Quantised langevin stochastic dynamics for bayesian federated
  learning.
\newblock In {\em International Conference on Artificial Intelligence and
  Statistics}, pages 6459--6500. PMLR, 2022.

\bibitem[WR20]{wu2020coordinate}
Changye Wu and Christian~P Robert.
\newblock Coordinate sampler: a non-reversible gibbs-like mcmc sampler.
\newblock {\em Statistics and Computing}, 30(3):721--730, 2020.

\bibitem[YAG{\etalchar{+}}19]{yurochkin2019bayesian}
Mikhail Yurochkin, Mayank Agarwal, Soumya Ghosh, Kristjan Greenewald, Nghia
  Hoang, and Yasaman Khazaeni.
\newblock Bayesian nonparametric federated learning of neural networks.
\newblock In {\em International Conference on Machine Learning}, pages
  7252--7261. PMLR, 2019.

\bibitem[YM20]{yuan2020federated}
Honglin Yuan and Tengyu Ma.
\newblock Federated accelerated stochastic gradient descent.
\newblock {\em Advances in Neural Information Processing Systems},
  33:5332--5344, 2020.

\bibitem[ZBKM18]{zhang2018advances}
Cheng Zhang, Judith B{\"u}tepage, Hedvig Kjellstr{\"o}m, and Stephan Mandt.
\newblock Advances in variational inference.
\newblock {\em IEEE transactions on pattern analysis and machine intelligence},
  41(8):2008--2026, 2018.

\bibitem[ZHZ21]{zhu2021data}
Zhuangdi Zhu, Junyuan Hong, and Jiayu Zhou.
\newblock Data-free knowledge distillation for heterogeneous federated
  learning.
\newblock In {\em International Conference on Machine Learning}, pages
  12878--12889. PMLR, 2021.

\bibitem[ZLZ{\etalchar{+}}19]{zhang2019cyclical}
Ruqi Zhang, Chunyuan Li, Jianyi Zhang, Changyou Chen, and Andrew~Gordon Wilson.
\newblock Cyclical stochastic gradient mcmc for bayesian deep learning.
\newblock {\em arXiv preprint arXiv:1902.03932}, 2019.

\end{thebibliography}

\appendix

\section{Event time simulation}
\label{sec:switching-time}

An important practical aspect of PDMP simulation is drawing the random times $\tau$ satisfying Equation (2) of the manuscript. We wish to simulate $\tau$ such that
\[ \P(\tau \ge t) = \exp \left( -\int_0^t \lambda(\phi(s;x,v)) \right)\]
where $(x,v)$ is any initial position of the trajectory.

To achieve this using the method of Poisson thinning, we assume that for every $(x,v)$ there is a bounding function $t \mapsto \overline{\lambda}(t;x,v)$ such that $\lambda(\phi(t;x,v))) \leq \overline{\lambda}(t; x,v)$ for all $t \geq 0$. We furthermore assume that the functions $\overline{\lambda}(\cdot,x,v)$ are chosen in such a way that there is an explicit formula for the inverse function 
\begin{equation} \label{eq:H} H(y;x,v) := \inf \left\{ t \geq 0 :  \int_0^t \overline{\lambda}(s;x,v) \ d s  \geq y \right\}.\end{equation}
Now if $V \sim \mathrm{Uniform}[0,1]$, then $\sigma := H(-\log V;x,v)$ satisfies 
\[ \P(\sigma \geq t) = \exp\left( - \int_0^t \overline{\lambda}(s;x,v) \ d s \right).\]
In words, $\sigma$ is distributed according to the first jump time of a inhomogeneous Poisson process with rate function $(\overline{\lambda}(t;x,v))_{t \geq 0}$.

In order to obtain a switching time with the desired distribution, we follow an iterative procedure.
We sample a proposed switching time $\tau$ satisfying $\P(\tau \geq t) = \exp(-\int_0^t \overline{\lambda}(s;x,v) \ d s)$, which we accept as true switching time with probability $\lambda(\phi(\tau;x,v))/\overline{\lambda}(\tau;x,v)$. If we do not accept the proposed switching time, we increase the time variable by $\tau$ and repeat with new starting point $\phi(\tau;x,v)$. The full procedure is given in Algorithm~\ref{alg:event-time-general}. 

\begin{algorithm}[ht!]
\renewcommand{\algorithmicrequire}{\textbf{Input:}}
\renewcommand{\algorithmicensure}{\textbf{Output:}}
\begin{algorithmic}[1]
\REQUIRE Current position $(x,v)$, switching rate function $t \mapsto \lambda(\phi(t;x,v))$ bounded from above by $\overline{\lambda}(\cdot;x,v)$ with associated function $H(\cdot;x,v)$ satisfying~\eqref{eq:H}.
\ENSURE Switching time $\tau$ such that $\mathbb P(\tau \geq t) = \exp \left( -\int_0^t \lambda(\phi(s;x,v))  \ d s\right)$.
\STATE Set $S$ = {\bf false}, $\tau = 0$.
\WHILE {{\bf not} $S$}
\STATE Simulate $V \sim \mathrm{Uniform}[0,1]$
\STATE Set $\sigma = H\left(-\log V; x, v \right)$
\STATE Set $S$ = {\bf true} with prob. $\lambda (\phi(\sigma;x,v)) / \overline{\lambda}(\sigma;x,v)$ 
\STATE Set $\tau = \tau + \sigma$ and $(x,v) = \phi(\sigma;x,v)$
\ENDWHILE
\end{algorithmic}
\caption{Event time simulation}
\label{alg:event-time-general}
\end{algorithm} 

In many settings we have that $\lambda(x,v) = \left( \langle v, \nabla U(x) \rangle\right)_+$, (or its one-dimensional variant, $\lambda_i(x,v) = (v_i \partial_i U(x))_+$), where $U$ has a bounded Hessian. If we also assume linear trajectories 
\[ \phi(t;x,v) = (x + vt, v),\]
as used in the  Zig-Zag Sampler and the Bouncy Particle Sampler, using Lipschitz continuity of $a \mapsto (a)_+$, we find using the mean value theorem that
\begin{align*} \lambda(\phi(t;x,v)) & = \lambda(x+vt, v) \le \lambda(x,v) + \sup_{x'} \| \nabla^2 U(x')\| \|v\|^2 t, \\
\lambda_i(\phi(t;x,v)) & = \lambda_i(x+vt, v) \le \lambda_i(x,v) + \sup_{x'} \| \nabla^2 U(x')\|_p \|v\|_p |v_i|  t,\end{align*}
where $\|\cdot\|$ denotes the Euclidean norm, and we can use any vector norm  $\|\cdot \|_p$, and associated induced matrix norm $\|\cdot\|_p$, for $ p \in [1,\infty]$, in the estimate for $\lambda_i$.

In the context of Federated PDMC as discussed in Section 3 the procedure outlined in Algorithm~\ref{alg:event-time-general} can be used for the simulation of the switching times $\tau_m$ for each machine: in this case we just replace the global switching rate $\lambda$ by the switching rates $\lambda_m$ of the individual machines and make sure we find a suitable upper bound $\overline{\lambda}_m$ with accompanying inverse $H_m$.

\section{Expected switching rate for exponential families}

In this section we consider the expected switching rate for the federated one-dimensional Zig-Zag sampler for data $(y_i)$  generated from an exponential family model with parameter $x$,
\begin{equation} \label{eq:exponential-family} f(y; x) = h(y) \exp \left( c(x) + x t(y) \right).\end{equation} 
We consider the situation where the data $(y_1, \dots, y_N)$ is generated from~\eqref{eq:exponential-family} for a fixed `true' parameter $x_0$.
We assume that the data $(y_1,\dots,y_N)$ is partitioned into $M$ batches $m =1,\dots, M$, where the $m$th batch consists of $n_m$ elements, denoted by $y_{m,i}$, $i = 1,\dots, n_m$.
For simplicity we assume a flat (Lebesgue) prior measure. 

Observe that
$U'(x) = \sum_{m=1}^M  U_m'(x)$
where 
\begin{align*} U_m'(x) & =  - n_m  c'(x) - \sum_{i=1}^{n_m}   t(y_{m,i}) = \frac{n_m}{N} U'(x) + Z_m
\end{align*}
having  defined
\[ Z_m := \frac {n_m}{N} \sum_{i=1}^N  t(y_i)- \sum_{i=1}^{n_m} t(y_{m,i}),\]
and using that
\[ U'(x) = \sum_{m=1}^M \sum_{i=1}^{n_m} \left[ - c'(x) -t(y_{m,i}) \right] = - N c'(x) - \sum_{i=1}^{N} t(y_i).\]

We consider the Zig-Zag Sampler in $\R$. Write $\mu(dx,dv) \propto \exp(-U(x)) \, dx \otimes \text{Uniform}(\{-1,+1\})$.
For the federated learning switching intensity we estimate, using the 1-Lipschitz property of $a \mapsto (a)_+$,
\begin{align} \nonumber \E_{\mu} \lambda_{\text{fed}}(x,v) & = \E_{\mu} \sum_{m=1}^M (v U_m'(x))_+  \le \E_{\mu} \sum_{m=1}^M \frac{n_m}{N} ( vU'(x) )_+ + \E_{\pi} \sum_{m=1}^M \left| Z_m \right|  \\
\label{eq:federated-rate-breakdown} & =  \E_{\mu} \lambda_{\text{can}}(x,v) + \sum_{m=1}^M  \left|Z_m \right|.
\end{align}
We analyze the two terms separately. 

First we consider the posterior expectation of the canonical switching rate, $ \E_{\mu} \lambda_{\text{can}}(x,v)$.
We have $\E_{\pi} U'(x) = 0$. Assuming posterior contraction (see the Bernstein-von Mises theorem, \cite[Section 10.2]{Vaart1998}), we have asymptotically that 
\begin{equation}\label{eq:asymptotic-contraction} \pi \stackrel{N \rightarrow \infty}{\sim} \mathcal N\left(\hat x, \frac 1 {N I(x_0)}\right), \end{equation}
where $\hat x = \hat x(y_1, \dots, y_N)$ denotes the maximum likelihood estimator for $x$ and $I(x)$ denotes the Fisher information associated with the parametric model $f(\cdot; x)$. (The expression~\eqref{eq:asymptotic-contraction} should be interpreted in an appropriate asymptotic sense as in \cite{Vaart1998}.) By the Delta method, asymptotically,
\[ \Var_{\pi} U'(x) = N^2 \Var_{\pi} c(x) = N^2  \left( \frac 1 N \frac{ \left( c'(\hat x) \right)^2}{I(x_0)} \right),  \]
so that 
\[  \E_{\mu} \lambda_{\text{can}} \le  \E_{\pi} |U'(x)| \le \left(  \E_{\pi} |U'(x)|^2 \right)^{1/2} = \sqrt{N} |c'(\hat x)| / \sqrt{I(x_0)}.\]

Next we consider the effective excess switching rate, bounded by $\sum_{m=1}^M |Z_m|$. We have
\begin{align*}
Z_m = \frac{n_m}{N}\sum_{i=1}^N t(y_i) - \sum_{i=1}^{n_m} t(y_{m,i}) &= \frac {n_m}{N} \sum_{\substack{m'=1 \\ m' \ne m}}^M \sum_{i=1}^{n_{m'}} t(y_{m',i}) + \sum_{i=1}^{n_m} \left(\frac{n_m}{N}- 1 \right) t(y_{m,i}).
\end{align*}
Recall that $y_i$ are i.i.d. according to $f(\cdot;x_0)$ for a fixed parameter $x_0$,
and denote $\sigma_0^2$ for the variance of $t(y_i)$. We see that the mean of $Z_m$ is zero, and its variance is 
\[ \sigma_0^2 \left( \frac {n_m^2} {N^2} \left( N - n_m \right)  + \left( \frac {n_m}{N} - 1\right)^2 n_m \right) = \frac {\sigma_0^2 n_m(N-n_m)}{N}. \]
Therefore we have
\begin{align*}
\E_y \sum_{m=1}^M \left| Z_m \right|
& \le \sum_{m=1}^M \sqrt{\Var_y {Z_m}} =  \sigma_0 \sum_{m=1}^M \sqrt{\frac{n_m (N-n_m)}{N} } \le \sigma_0 \sum_{m=1}^M \sqrt{n_m}. 
\end{align*}
Using Jensen's inequality,
\begin{align*}
\sigma_0 \sum_{m=1}^M \sqrt{n_m} = \sigma_0 \sum_{m=1}^M \frac{n_m}{N} \frac{N}{n_m} \sqrt{n_m} \le \sigma_0 \left( \sum_{m=1}^M \frac{n_m}{N} \frac{N^2}{n_m} \right)^{1/2}  = \sigma_0 \sqrt{M N}.
\end{align*}

Combining the terms in~\eqref{eq:federated-rate-breakdown} we find that
\[ \E_{\mu} \lambda_{\text{fed}}(x,v) = \mathcal O( \sqrt{N})+ \mathcal O(\sqrt{N M } )  = \mathcal O(\sqrt{NM}).\]

According to this analysis the switching rate of the Federated PDMC Sampler (Algorithm 2 in the manuscript) is increased by a term of $\mathcal O(\sqrt{MN })$ relative to the canonical rate of the standard PDMC Sampler (Algorithm 1 in the manuscript), which is $\mathcal O(\sqrt{N})$.

\section{Differential privacy}

For simplicity we write the proof of Theorem 5.1 in the manuscript in terms of time dependent switching rates $\lambda(t)$. These may be interpreted as $\lambda(t) = \lambda(\phi(t;x,v))$ in the context of Federated PDMC.

\begin{lemma}
\label{lem:differential-privacy-bounded-difference}
Suppose $|\lambda(t) - \tilde \lambda(t)| \le K$ and $\frac 1 {\gamma} \le \lambda(t)/\tilde \lambda(t) \le \gamma$ for $t \ge 0$ and some constants $K > 0$ and $\gamma > 1$. Then
\[ \frac{\exp \left( -\int_0^t \lambda(s) \right) \lambda(t)}{\exp \left( -\int_0^t \tilde \lambda(s) \right) \tilde \lambda(t)} \le \exp \varepsilon \]
if $\log \gamma < \varepsilon$ and $ t \le \frac{\varepsilon - \log \gamma}{K}$.
\end{lemma}
\begin{proof}
We have 
\begin{align*}
    \frac{\exp \left( -\int_0^t \lambda(s) \right) \lambda(t)}{\exp \left( -\int_0^t \tilde \lambda(s) \right) \tilde \lambda(t)} & \le \exp \left( \int_0^t  |\lambda(s) - \tilde \lambda(s) | \, d s \right) \frac{\lambda(t)}{\tilde \lambda(t)} \le \gamma \exp (K t).
\end{align*}
The stated result follows immediately.
\end{proof}

\begin{proof}
We have
\begin{align*}
    \frac{\exp \left( -\int_0^t \lambda(s) \right) \lambda(t)}{\exp \left( -\int_0^t \tilde \lambda(s) \right) \tilde \lambda(t)} & = \exp \left( \int_0^t \lambda(s) \left[ \tilde \lambda(s)/\lambda(s) - 1 \right] \,  d s \right) \lambda(t)/\tilde \lambda(t) \\
    & \le \gamma \exp \left( \int_0^t (\alpha + \beta s) (\gamma - 1) \, d s\right)  =  \gamma \exp \left( (\gamma- 1) ( \alpha t + \tfrac 1 2 \beta t^2 ) \right),
\end{align*}
from which it is straightforward to obtain the stated result.
\end{proof}

We see that under reasonable conditions we have $\varepsilon$-differential privacy, provided that we restrict the maximal switching time to a finite time interval. As an extension, we may obtain an upper bound on the probability to have a switching time larger than this time, in order to obtain $(\varepsilon, \delta)$-differential privacy.

\begin{lemma}
\label{lem:epsdelta-diffprivacy-abstract}
Let $(E, \mu)$ be a measure space. Suppose $f$ and $\tilde f$ are probability densities relative to $\mu$. Suppose $G \subset E$ is such that $\int_{G^c} f \, d \mu \le \delta$ for some $\delta \ge 0$, and $f/\tilde f \le \exp(\varepsilon)$ on $G$. Then
\[ \int_S f \, d \mu \le \exp(\varepsilon) \int_S \tilde f\, d \mu + \delta.\]
\end{lemma}
\begin{proof}
This follows since
\[ \int_S f \, d \mu = \int_{S \cap G} f \, d \mu + \int_{S \cap G^c} f \, d \mu \le \exp(\varepsilon) \int_{S \cap G} \tilde f\, d \mu + \delta. \]
\end{proof}

We are now ready to provide the proof of Theorem 5.1.

\begin{proof}
Using Lemma~\ref{lem:epsdelta-diffprivacy-abstract} it suffices to show that, for some $t_0 > 0$, 
\[ \P(\tau \ge t_0) \le \delta \quad \text{and} \quad f_{\tau}(t)/ f_{\tilde \tau}(t) \le \exp(\varepsilon), \quad 0 \le t \le t_0,\]
where $f_{\tau}$ denotes the density function of $\tau$.
Let $t_0 = \frac{\varepsilon - \log \gamma}{K}$ with $\gamma = \frac K \rho + 1$.
Indeed we have
\[ \P(\tau \ge t_0) \le \exp(-t_0 \rho) = \delta.\]
Furthermore we have that 
\[ \lambda(t) / \tilde \lambda(t) \le \frac{\lambda(t) - \tilde \lambda(t)}{\tilde \lambda(t)} + 1 \le 1 + \frac{K}{\rho} = \gamma\]
and using Lemma~\ref{lem:differential-privacy-bounded-difference} we find that the ratio of densities $f_{\tau}/f_{\tilde \tau}$ is bounded by $\exp(\varepsilon)$ for $t \le t_0$.

Finally, for fixed $\varepsilon > 0$ and $\delta > 0$, if we take $\rho$ satisfying the indicated inequality, then
we must verify that $\varepsilon > \log \left( 1 + \frac K {\rho} \right)$ and that
\[ \delta' := \exp \left( - \frac{\rho}{K} \left[ \varepsilon - \log \left( 1 + \frac{K}{\rho} \right) \right]\right) \le \delta.\]
Indeed,
\[ \log \left( 1 + \frac K {\rho} \right) \le \left( 1 + \frac{\varepsilon}{1 + \log (1/\delta)} \right) < \log (1 + \varepsilon) < \varepsilon,\]
and
\[ \log \delta' = \frac{\rho}{K} \left[ \log \left( 1 + \frac{K}{\rho} \right) - \varepsilon \right] \le \frac{\rho}{K} \left( \frac{K}{\rho} - \varepsilon \right) \le \log \delta. \]
\end{proof}

\section{Increased Privacy via Dynamic Prior Switching}
Suppose the communication between the server and one the nodes is compromised by an attacker.  Then over a long period of time, it would be theoretically possible for the attacker to reconstruct the global potential function by observing the switches over long periods of time.   To mitigate this we propose a heuristic strategy which obfuscates the individual worker contributions to the likelihood by weighting them with a random piece of the prior, which evolves dynamically.   To be more specific, suppose we are targeting the following global posterior:
$$
    \pi(x) \propto \exp\left(- \sum_{m=1}^M U_m(x)\right)\exp(-U_0(x)),
$$
where $U_m$ is the potential contribution for the $m^{th}$ worker, and $\pi_0 \propto e^{-U_0}$ is a prior.   The proposed strategy is to distribute the prior across the workers so that the local potential for the $m^{th}$ worker becomes 
$$
    \widetilde{U}_m(x) = U_m(x) + \alpha_m U_0(x),
$$
where the weights $\alpha_1, \ldots, \alpha_M$ are randomly chosen scalars such that $\sum_{m=1}^M \alpha_m = 1$.   Clearly, $\exp(-\sum_m \widetilde{U}_m(x)) \propto \pi(x)$.    To dynamically change the weights, we introduce a constant redistribution rate $\lambda_{\text{redist}}$ and assume that prior re-distributions (i.e. resampling of the $\alpha$'s) occurs at discrete times determined by an independent Poisson process with constant rate $\lambda_{\text{redist}}$.   The new algorithm is expressed in Algorithm \ref{alg:redist}.  Note that the worker routine remains unchanged from the standard Federated Zig-Zag algorithm.  

\begin{algorithm}[ht!]
\renewcommand{\algorithmicrequire}{\textbf{Input:}}
\renewcommand{\algorithmicensure}{\textbf{Output:}}

\begin{algorithmic}[1]
\REQUIRE Initial condition $(x, v) \in \R^d \times \mathcal V$, redistribution rate $\lambda_{\text{redist}} > 0$\\
\ENSURE The sequence of \emph{skeleton points} $(T_k, X_k, V_k)_{k=0}^{\infty}$.
\STATE Set $(T_0, X_0, V_0)=(0, x, v)$.
\FOR{$k=0,1,2,\ldots$}
	\STATE{\emph{Locally}: Every worker simulates $\tau_m$ and $\tilde V^{(m)}$ such that $$ \P(\tau_m \geq t)  = \exp \left( -\int_0^t \lambda_m(\phi(s;X_k,V_k)) \ d s \right),
    $$
    and $\tilde V^{(m)}  \sim Q_m(\phi(\tau_m;X_k,V_k),\cdot)$,  where $$\lambda_m(x,v) = \max(0, U_m(x) + \alpha_m U_0(x)).$$}
		\STATE{\emph{Centrally}: \\ 
        Simulate $\tau_{\text{redist}}$ such that 
        $$
            \mathbb{P}(\tau_{\text{redist}} \geq t)  = \exp\left(-\lambda_{\text{redist}} t\right).
        $$
    Set  
        \begin{align*} m_0 &= \argmin \{ \tau_1, \dots, \tau_M \}, \\
                       \tau_{min} &= \min(\tau_{m_0}, \tau_{\text{redist}})\\
	               T_{k+1} &= T_{k} + \tau_{min},\\
                    X_{k+1} &= \phi(\tau_{min}; X_k, V_k),
        \end{align*}}
          \IF{$\tau_{\text{redist}} < \tau_{m_0}$}
            \STATE{Resample $\alpha_1, \ldots, \alpha_M$, such that $\sum_{m=1}^M \alpha_m = 1$.}
         \ELSE
         \STATE{Set $$
            V_{k+1} = \tilde V^{(m_0)}$$}
        \ENDIF
\ENDFOR
\end{algorithmic}

\caption{Federated ZigZag with Prior Re-distribution}
\label{alg:redist}
\end{algorithm}

The dynamic re-distribution of the prior introduces time-inhomogeniety into the process, and it does not formally follow from previous results that the resulting process has the correct unique stationary distribution, though it is intuitively clear that would be the case.   We leave the analysis of this new PDMP variant as a subject for future work.   To demonstrate that the redistribution process does not  affect convergence to equilibrium, we repeat the experiment from the Logistic Regression example from Section 6.2 of the main paper, under the same conditions, with $\lambda_{\text{redist}}=0.1$.   We plot the marginal Wasserstein distances against a HMC reference sample in Figure \ref{fig:perf_lr_error_dist}.  Compared to the results for the Federated Zig-Zag in Figure 3 of the main text we observe that the introduction of the prior re-distribution process does not adversely affect the convergence behaviour of the continuous-time process.

\begin{figure}
\begin{center}
 \includegraphics[width=\figurewidth]{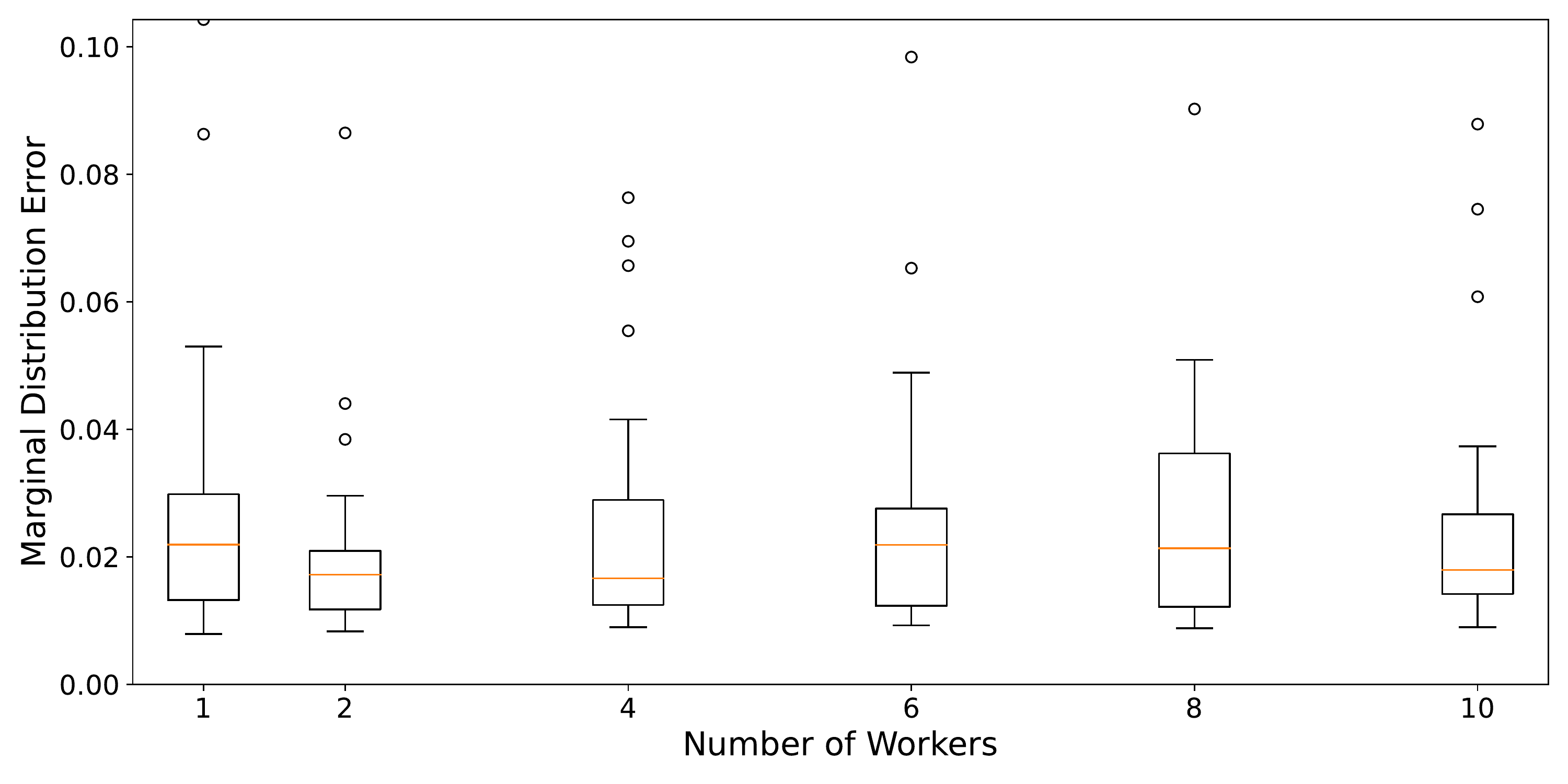}
 \caption{Marginal Wasserstein distance against a reference HMC sample for the logistic regression model using Federated ZigZag with dynamic prior re-distribution.} 
 \label{fig:perf_lr_error_dist}
\end{center}
\end{figure}

\section{Further Numerical Experiments}

\subsection{Log-Gaussian Cox Model}

To demonstrate a more complex example of federation across multiple workers, we consider a Log-Gaussian Cox model, similar to the one considered in \cite{galbraith2016event} and \cite{wu2020coordinate} for the Zig-Zag and Coordinate samplers, respectively.  We assume that the observations ${y} = \lbrace y_{ij}\rbrace_{i, j=1}^d$ are  Poisson distributed and conditionally independent given the intensity ${\lambda} = \lbrace \lambda_{ij}\rbrace_{i,j=1}^d$ where $\lambda_{ij} = \exp(x_{ij})$.  The latent process ${x} = \lbrace x_{ij}\rbrace_{i,j=1}^d$ is a Gaussian 
process defined on the $d\times d$ grid with mean zero and precision matrix $P_{uv} = \beta(\delta_{u,v} - \alpha A_{u, v})$, where ${A}=(A_{u, v})$ is the adjacency matrix of the grid and where $u=(i,j)$ and $v=(i', j')$ are grid coordinates of the nodes.     The posterior distribution for ${x}$ given the observations ${y}$ is given by
$$
\pi({x}) = \mathbb{P}({x} \, | \, {y}, \alpha, \beta) \propto \exp\left[\sum_{i,j=1}^d \left(y_{ij}x_{ij} - \exp(x_{ij})\right) - \frac{\beta}{2}{x}^\top ({I} - \alpha {A}) {x} \right]
$$
We assume that the nodes are distributed spatially across $M=4$ workers, see Figure \ref{fig:cox1}.  Let $V_1, \ldots, V_M$ be the nodes assigned to each worker.  The $m^{th}$ worker will target local potential:
$$
    U_m({x}) = \sum_{(i,j) \in V_m} \exp(x_{ij}) - y_{ij}x_{ij}, \quad m=1,\ldots, M.
$$
The prior term is handled by the central server which has a potential of the form
$$
    U_0({x}) = \frac{\beta}{2}{x}^\top ({I} - \alpha {A}) {x}.
$$
Note that the $m$-th worker only needs to observe the latent variables $x_{ij}$ for $(i,j) \in V_m$ which are relevant to its observations.  The interactions between the latent variables through the Gaussian process prior are handled entirely the central server.   We choose $\alpha = 0.1$ and $\beta = 1$.  For clarity of presentation we demonstrate the scheme for $d = 4$, i.e. a grid with $16$ nodes, noting that our experiments run fine on much larger grids.   We run the Federated ZigZag sampler for the $4$ workers until time $T=100.0$.  We extract samples from the  continuous time process by extracting samples from the resulting trace at $10^{-3}$ time-steps obtaining approximately $10^5$ samples. In Figure \ref{fig:cox_results} we plot the marginal posteriors for the output of the Federated ZigZag process, compared against a reference sample generated by running HMC on the same target distribution.    We see excellent agreement between the two sets of distributions. 
\begin{figure}
\begin{center}
 \includegraphics[width=.8 \figurewidth]{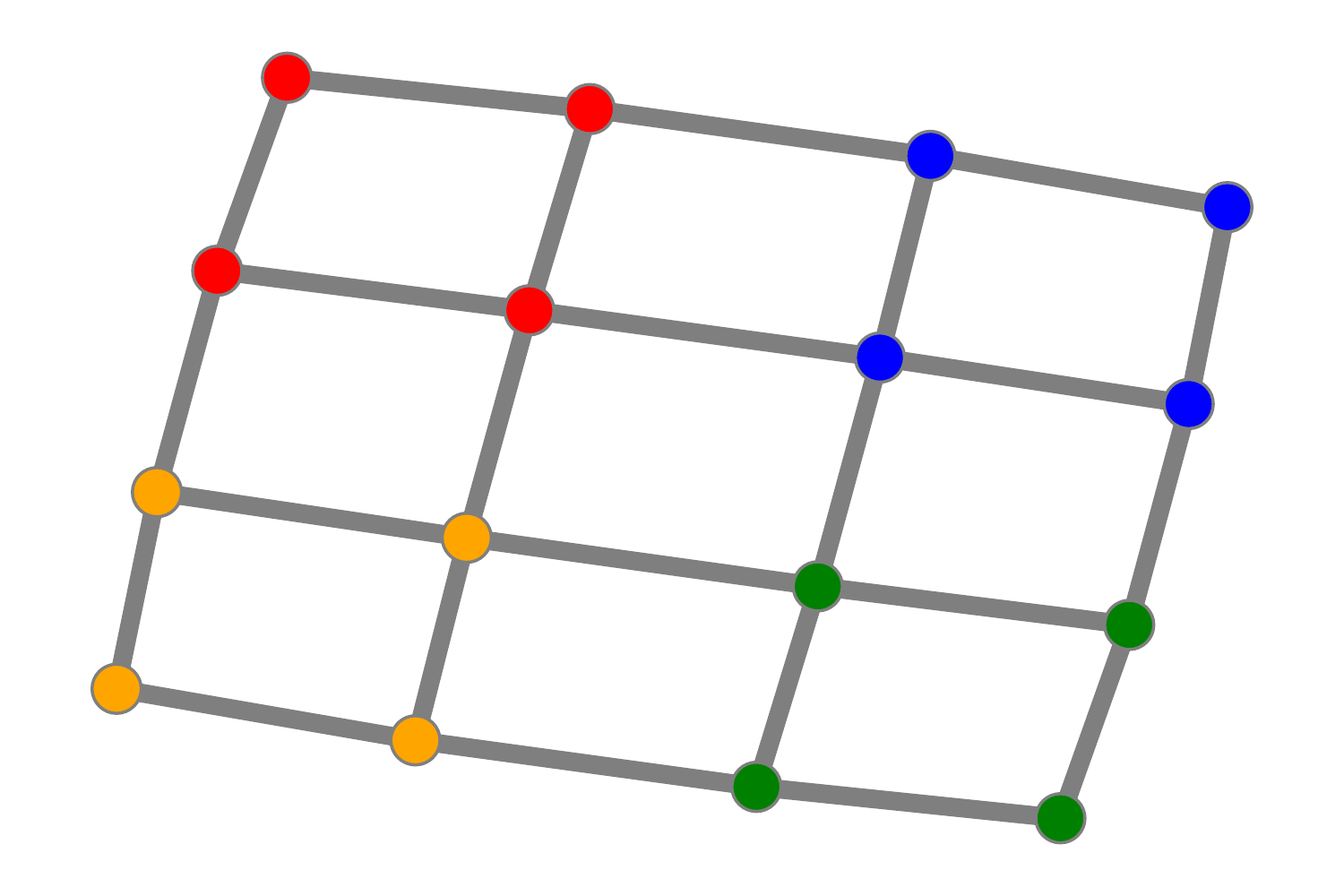}
 \caption{Graph underlying the Log-Gaussian Cox model, where the nodes are partitioned across the $4$ different workers (distinguished by node colour).}
 \label{fig:cox1}
\end{center}
\end{figure}

\begin{figure}
\begin{center}
 \includegraphics[width=.8 \figurewidth]{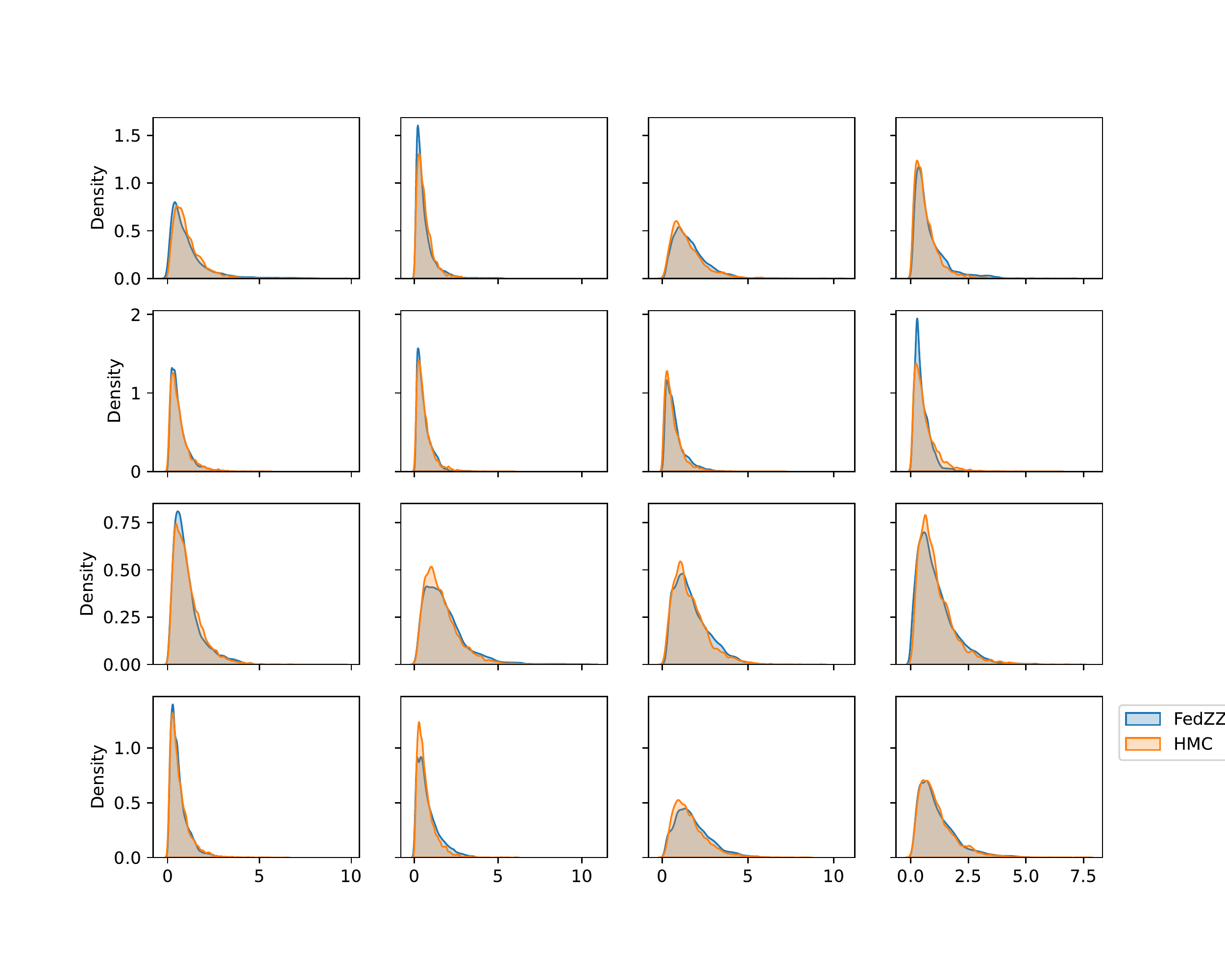}
 \caption{Marginal posteriors for the Log-Gaussian Cox model, Federated ZigZag sampler compared to a HMC reference sample.}
 \label{fig:cox_results}
\end{center}
\end{figure}

\section{Specification of Example Models}
\subsection{Multivariate Gaussian Distribution}

Consider a multivariate Gaussian distribution $\mathcal{N}(\mu, \Sigma)$ where $\mu \in \mathbb{R}^d$ and $\Sigma$ is a symmetric, positive definite $d\times d$ covariance matrix.   Then the associated potential is given by
$$
    U(x) = \frac{1}{2} (x-\mu)^\top P  (x-\mu), \quad x \in \mathbb{R}^d,
$$
where $P = \Sigma^{-1}$, and the intensity functions for the Zig-Zag process take the form
\begin{align*}
    \lambda_k(x, v) = \max\left(0, v_k \sum_{i=1}^d P_{ki}(x_i-\mu_i)\right),
\end{align*}
so that
\begin{align*}
    \lambda_k(x+tv, v) &= \max\left(0, v_k \sum_{i=1}^d P_{ki}(x_i + t v_i -\mu_i)\right)\\
    &=\max\left(0,  \sum_{i=1}^d v_k P_{ki}(x_i - \mu_i)+ t \sum_{i=1}^d v_k P_{ki}v_i)\right).
\end{align*}
Given that this intensity is of the form $\max(0, b + at)$ the distribution of the next event can be sampled directly by computing the associated inverse function $H(u; x, v)$ in \eqref{eq:H} exactly,   
so that for $V \in U[0,1]$, the random variable $H(-\log V; x, v)$ will be distributed according to the first jump time of the Poisson process with rate function $\max(0, b + at)$.  
\subsection{Logistic Regression}

\label{sec:logistic-regression}

We assume that worker $m$, with $m \in \{1, \dots, M\}$, has access to independent observations $(\xi^{(m)}_i, \eta^{(m)}_i)_{i=1}^{n_m} \subset \R^d \times \{0,1\}$ from the joint model 
\[ \P(\eta_i^{(m)} = 1 \mid x) = \frac 1 { 1+ \exp \left(- x^T \xi_i^{(m)}\right)}, \quad i = 1, \dots, n_m, \quad m = 1, \dots, M.\]
All observations are mutually independent, also between different workers.
Here $x \in \R^d$ is the parameter which we wish to infer. We assume a prior distribution $\pi_0(x) \propto \exp(-U_0(x))$ over the unknown parameter $x$.

The posterior distribution is then specified as
\begin{align} \pi(x) \propto  \pi_0(x) \prod_{m=1}^M \prod_{i=1}^{n_m} \frac{\exp \left( \eta_i^{(m)}  x^T \xi_i^{(m)} \right)}{1 + \exp \left(   x^T \xi_i^{(m)} \right)}.\end{align}
It can be written as $\pi(x) \propto \exp(-U(x))$, where
\begin{align*}
    U(x) & = U_0(x) + \sum_{m=1}^M \sum_{i=1}^{n_m}  \left\{ \log \left[ 1 + \exp \left(   x^T \xi_i^{(m)} \right) \right]  - \eta_i^{(m)}  x^T \xi_i^{(m)} \right\} = \sum_{m=1}^M U_m(x),
\end{align*}
where
\[ U_m(x) =  \frac {n_m}{N} U_0(x) + \sum_{i=1}^{n_m} \left\{ \log \left[ 1 + \exp \left(   x^T \xi_i^{(m)} \right) \right]  - \eta_i^{(m)}  x^T \xi_i^{(m)} \right\} 
\]
with $N = \sum_{m=1}^M n_m$.

The gradients of $U_m$ are given by
\[ \nabla U_m(x) = \frac{n_m}{N} \nabla U_0(x) + \sum_{i=1}^{n_m} \xi_i^{(m)}  \frac{1-\eta_i^{(m)}(1 + \exp(-x^T \xi_i^{(m)}))}{1 + \exp(-x^T \xi_i^{(m)})}\]
and the associated Hessians
\begin{align*}  \nabla^2 U_m(x) 
& = \frac{n_m}{N} \nabla^2 U_0(x) + \sum_{i=1}^{n_m}  \frac {\xi_i^{(m)}(\xi_i^{(m)})^T}{4 \cosh(x^T \xi_i^{(m)})^2}.
\end{align*}
In particular, the Hessian of $U_m$ is bounded for each $m$  provided that the Hessian of $U_0$ is bounded.

\subsection{Time Series Model}

Write  $h(z) = -\log g(z) = \left( \frac{\nu+1}{2} \right) \log \left( 1 + \frac {z^2}{\nu} \right)$. We have 
\[\pi(x,c) \propto \exp \left(-\sum_{i=1}^N \sum_{k=1}^K U_{i,k}(x,c) \right) \quad \text{with} \quad 
U_{i,k}(x,c) = h(x y_{k-1}^{(i)} + c - y_k^{(i)}).\]
We have
\begin{align*}
h'(z) & = \frac{(\nu+1)z}{\nu + z^2}, \quad h''(z) = \frac{(\nu+1)(\nu - z^2)}{(\nu+z^2)^2},
\end{align*}
which admit uniform bounds
\begin{align*}
|h'(z)| & \le  \frac{(\nu+1)\sqrt{\nu}}{2 \nu^2},\\
|h''(z)| & \le \frac{\nu+1}{\nu}.
\end{align*}
By the chain rule,
\begin{align*} \nabla U_{i,k}(x,c) & =  h'(x y_{k-1}^{(i)} + c - y_k^{(i)})
 \begin{pmatrix} y_{k-1}^{(i)} \\ 1 \end{pmatrix},
\end{align*}
and
\begin{align*}
\nabla^2 U_{i,k}(x,c) & = h''(x y_{k-1}^{(i)} + c - y_k^{(i)}) \begin{pmatrix} (y_{k-1}^{(i)})^2 & y_{k-1}^{(i)} \\ y_{k-1}^{(i)} & 1 \end{pmatrix}
\end{align*}
We may therefore obtain a bound on the Hessian norm as
\[ \sup_{x,c}\| \nabla^2 U(x,c) \| \le \frac{\nu+1}{\nu} \sum_{i=1}^N \sum_{k=1}^K (1 + y_{k-1}^{(i)})^2. \]

\subsection{Log Gaussian Cox Model}

The potential function for a single worker is given by 
$$
U_m(x) = \sum_{(i,j) \in V_m} (-y_{ij}x_{ij} + \exp(x_{ij})),
$$
so that
$$
    \nabla_{x_{ij}} U_m(x) = -y_{ij} + \exp(x_{ij}),  \mbox{ for } (i,j) \in V_m.
$$
For simplicity we flatten the index, so that $k = d(i-1) + j$, and re-index $x_k = x_{(i,j)}$, so that
$$
    \nabla_{x_{k}}U_m(x) = -y_k + \exp(x_k).
$$
The associated switching intensity is given by
$$
\lambda_k(x + tv, v) = \left(v_k \nabla_k U_m(x + tv)\right)_{+}.
$$
  We can bound this above as follows:
\begin{align*}
\lambda_k(x + tv, v) &= \left( -y_k v_k + v_k \exp(x_k + t v_k) \right)_{+} \\
&\leq \left( -y_k v_k  \right)_{+} + \left( v_k \exp(x_k + t v_k) \right)_{+}
\end{align*}

The last upper bound provides a means of exactly simulating a dominating inhomogeneous Poisson process which can be subsequently thinned to simulate the local worker Zig-Zag process for the Log-Gaussian Cox model.   Indeed, we can perform the following steps to simulate the next switching local Zig-Zag sampler event from $(x,v)$.  
\begin{enumerate}
\item Simulate the next event of the Poisson process with intensity  $\max(0, -y_k v_k)$,  i.e. let $u \sim U[0,1]$, then set 
    $$
        \tau_k^{(1)} = \begin{cases} 
                        \frac{\log(u)}{y_k v_k} & \mbox{if } y_k v_k < 0 \\
                        \infty &\mbox{ otherwise.}
                \end{cases}
    $$
\item Simulate the next event of the inhomogenous Poisson process with intensity $\max(0, v_k \exp(x_k + t v_k))$, i.e. let $u' \sim U[0,1]$, then set 
$$
    \tau_k^{(2)} = \begin{cases}
                \log(\exp(x) - \log(u')) - x & \mbox{ if } v_k > 0 \\
                \infty & \mbox{ otherwise.}
             \end{cases}
$$
\item Set $\tau_k = \min(\tau_k^{(1)}, \tau_k^{(2)})$, for $k=1,\ldots, $ 
\item Let $m_0 = \arg\min \tau_k$.
\item Set $x = x + \tau_{m_0} v$.
\item Let $T_1 = -y_l v_l$ and $T_2 = v_k\exp(x_l)$. With probability $\frac{\max(0, T_1 + T_2)}{\max(0, T_1) + \max(0, T_2)}$ communicate a switch to the central server, otherwise, return to step $1$.
\end{enumerate}

\end{document}